\documentclass[11pt]{article}

\usepackage{graphicx}
\usepackage{dcolumn}
\usepackage{bm}

\usepackage[,
text={7in,10in},
]{geometry}
\usepackage[utf8]{inputenc}
\usepackage[T1]{fontenc}
\usepackage[english]{babel}
\usepackage{float}
\usepackage{gensymb} 
\usepackage{mathrsfs} 
\usepackage{braket}
\usepackage{authblk}
\usepackage{hyperref}
\usepackage{amssymb}
\usepackage{amsthm}
\usepackage{braket}
\usepackage{amsmath}
\usepackage{setspace}
\newtheorem{theorem}{Theorem}
\newtheorem{lemma}{Lemma}
\newtheorem{definition}{Definition}
\newtheorem{proposition}{Proposition}
\newtheorem{corollary}{Corollary}

\usepackage{color}

\begin{document}

\title{Efficient approximate unitary $t$-designs from $partially$ $invertible$ $universal$ sets and their application to  quantum speedup.}

\author{Rawad Mezher$^1$  $^2$} 
 
\author{Joe Ghalbouni$^2$}
\author{Joseph Dgheim$^2$}
\author{Damian Markham$^1$}
\affil{
(1)Laboratoire d'Informatique de Paris 6, CNRS, Sorbonne Université, 4 place Jussieu, 75252 Paris Cedex 05, France\\
(2) Laboratoire de Physique Appliquée, Faculty of Sciences 2, Lebanese University, 90656 Fanar, Lebanon \\ \href{rawad.mezher@etu.upmc.fr}{ rawad.mezher@etu.upmc.fr} \\
\href{joe.ghalbouni@ul.edu.lb}{joe.ghalbouni@ul.edu.lb} \\ \href{jdgheim@ul.edu.lb}{jdgheim@ul.edu.lb} \\ \href{damian.markham@lip6.fr}{ damian.markham@lip6.fr}  \\}
 
%
%

\maketitle
\begin{abstract}

 At its core a $t$-design is a method for sampling from a set of unitaries in a way which mimics sampling randomly from the Haar measure on the unitary group, with applications across quantum information processing and physics.

 We construct new families of quantum circuits on $n$-qubits giving rise to $\varepsilon$-approximate unitary $t$-designs efficiently in $O(n^3t^{12})$ depth. 
 These quantum circuits are based on a  relaxation of technical requirements in previous constructions.
 In particular, the construction of circuits which give efficient approximate $t$-designs by Brandao, Harrow and Horodecki \cite{brandao} required choosing gates from ensembles which contained inverses for all elements, and that the entries of the unitaries are algebraic.
 We reduce these requirements, to sets that contain elements without inverses in the set, and non-algebraic entries, which we dub \emph{partially invertible universal} sets.
 
 We then adapt this circuit construction to the framework of measurement based quantum computation (MBQC) and give  new explicit examples of $n$-qubit graph states with fixed assignments of measurements (graph gadgets) giving rise to unitary $t$-designs based on partially invertible universal sets, in a natural way. 
 
 We further show that these graph gadgets demonstrate a quantum speedup, up to standard complexity theoretic conjectures. 
 We provide numerical and analytical evidence that almost any assignment of fixed measurement angles on an $n$-qubit cluster state give efficient $t$-designs and demonstrate a quantum speedup.

\end{abstract}

\newpage
\section{Introduction}



The capacity to randomly choose a unitary operation is a powerful tool in quantum information and physics in general, with applications ranging from randomized benchmarking~\cite{benchmarking}, secure private channels~\cite{HLS+04} to understanding how quantum systems thermalize~\cite{muller}, as well as modeling black holes~\cite{haypres} and recently as providing natural candidates for devices demonstrating a quantum speedup~\cite{harrow2018,juan,juan2,mannbremner}. 
A $t$-design is an ensemble of unitaries and associated probabilities, which, when sampled, mimic choosing a unitary at random according to the Haar measure (the most natural group theoretical definition of random) in a specific sense - they act exactly as the Haar measure up to $t$th order in the statistical moments.
The main interest in $t$-designs lies in the fact that sampling from the Haar measure is known to require exponential resources \cite{knill}, but sampling from $t$-designs can be done efficiently~\cite{dankert,harrow2des,brandao,zhu,harrow2018,peterdamian,mezher,wintonviola,nakata}, whilst maintaining usefulness~\cite{benchmarking,HLS+04,muller,haypres,harrow2018,juan2,mannbremner}.

The prevalent technique for generating a $t$-design is through random circuits, where gates are randomly chosen from some ensemble of small, typically two qubit gates, and put together in a specific way to form a circuit~\cite{harrow2des,brandao,harrow2018,peterdamian,mezher,wintonviola}.
Though essentially any universal set of two qubit gates can be used to generate this ensemble, the precise conditions on this ensemble are somewhat strict (due to technical reasons in the proofs) - they require that each gate has an inverse in the ensemble and that their entries are algebraic.
The former condition is also imposed on universal ensembles when proving the Solovay Kiteav theorem for efficient approximate universality~\cite{solovaykitaev}.
Though usually this is not an issue, it can be, particularly when these sets of unitaries are generated in a restricted manner - for example arising from measurements on graph states~\cite{peterdamian,mezher,juan,juan2}.

Graph states are a family of multipartite quantum states, with simple descriptions in terms of graphs~\cite{graphstate}.
They are very useful resources for quantum information, with applications in measurement based quantum computing~\cite{briegelrauss}, fault tolerance~\cite{RHG07}, cryptographic multiparty protocols~\cite{Secret sharing Markham Sanders 08}, quantum networks~\cite{CDF} and recently for generating $t$-designs~\cite{peterdamian,mezher} and instances of quantum speedup~\cite{juan,juan2}.
They represent the cutting edge in terms of size of entangled states that can be generated and controlled in experiments, with implementations demonstrated in optics~\cite{WL+16,BK+12,BM+14}, \cite{CY+17,YU+13} including on chip~\cite{CO+16}, in ion traps~\cite{BM+11,MS+11}, super conducting qubits~\cite{SX+17} and NV centres~\cite{CK+16}. 
Remarkably, in continuous variable quantum optics graph states of up to $10^4$ parties have been created~\cite{YU+13}.
Furthermore there are several techniques that have been developed to verify the quality of graph states in various settings of trust \cite{juan,damianalexandra,gao,cramer,serfling} which can often be translated into verification of their applications.

Our work connects these different questions and approaches, first by proving a general relaxation of the conditions on a set of ensembles used to generate $t$-design, leading to new constructions for circuits, which we then translate to the graph state, measurement based approach.
We then give explicit examples where the relaxation to partially invertible sets is useful in graph state constructions.
Following along the lines of \cite{juan,juan2} we show that these examples give rise to natural instances of provably hard sampling problems demonstrating quantum speedup.

We now give a bit more background into the three areas of our main results.

\subsection{$t$-designs in partially invertible universal sets}

Exact $t$-designs, where the condition on the $t$th order moments are satisfied exactly (stated precisely in section 2.1),  are only known for a few cases \cite{dankert,zhu,peterdamian}. We are thus often interested in approximate versions, where conditions hold up to some error $\varepsilon$ - we call these $\varepsilon$-approximate $t$-designs \cite{harrow2des,brandao,harrow2018,peterdamian,mezher,wintonviola,nakata}. 
We say a circuit construction is efficient if the size of the circuit, $k$ does not scale exponentially in $n$, $t$ or $1/\varepsilon$. Previous work showed that random $n$-qubit quantum circuits formed of $k$ applications of 2-qubit gates form efficient $\varepsilon$-approximate $t$-designs with $ k=poly(n,t,log(\frac{1}{\varepsilon}))$ \cite{harrow2des,brandao}, where these 2-qubit gates are chosen from the Haar measure on $ U(4)$ \cite{brandao,harrow2des}, or uniformly randomly from a universal \footnote{A set $\mathcal{U} \subset U(N)$ is said to be universal in $U(N)$, when the group generated by $\mathcal{U}$ is dense in $U(N)$.}  set $\mathcal{U_{B} }\subset U(4)$ which contains unitaries and their inverses \footnote{ We mean by this that for every $U \in \mathcal{U_{B}}$, there exists $U_{1} \in \mathcal{U_{B}}$ , such that $U_{1}=U^{\dagger}$.}, and is  made up of unitaries with algebraic entries \cite{brandao}.

The first question we ask here is whether the restriction that $\mathcal{U_{B}}$ contains unitaries and their inverses can be avoided. 
Such a possibility opens up different constructions, which are notably important considering measurement based constructions, where one does not easily have full control over the whole ensemble.
The answer to this question, as we will show, turns out to be positive, provided $(I)$: we can find sets $\mathcal{U_{B}}$  containing a subset $\mathcal{U_{M}} \subset \mathcal{U_{B}}$ formed of unitaries with algebraic entries \cite{brandao}, such that $\mathcal{U_{M}}$ contains unitaries and their inverses, and $(II)$ both $\mathcal{U_{M}}$ and its complement in $\mathcal{U_{B}}$ - denoted by $\mathcal{U_{B/M}}$ - which need not nescessarily contain unitaries and their inverses nor have algebraic entries are universal in $U(4)$.  
For simplicity, we refer to sets $\mathcal{U_{B}}$  verifying $(I)$ and $(II)$ as \textit{partially invertible universal sets}. 

Based on this we derive a construction of $n$-qubit quantum circuits formed of blocks of 2-qubit unitaries chosen uniformly from $any$ partially invertible universal set in $U(4)$, and show that these circuits are efficient $\varepsilon$-approximate $t$-designs in depth $O(n^{3}t^{12})$. In our proofs, we use technical tools such as the Detectibility Lemma \cite{aharonov}, and techniques from \cite{harrow2des,brandao}.

We then adapt this circuit construction to a quantum computing paradigm where partially invertible universal sets arise quite naturally, namely measurement based quantum computation (MBQC) \cite{briegelrauss,raussbriepra}. In MBQC, measuring the non-output qubits of an $n$-qubit graph state at particular angles in the $XY$, $XZ$, or $YZ$ planes of the Bloch sphere, and performing a corrective strategy, for example given by the $g$-flow \cite{elham}, is sufficient to implement any unitary $U \in U(2^{n})$ on the $n$  unmeasured  output qubits. 
Recently it was shown that only $XY$ measurements on cluster states (graph states of the two dimensional square lattice) are sufficient for implementing any $U \in U(2^{n})$ \cite{mantri}. 
On the other hand, performing non-adaptive \footnote{Non-adaptive means with no corrective strategy, non-adaptive measurements can be performed simultaneously.} measurements on graph states effectively implements on the (unmeasured) output qubits unitaries sampled uniformly from an ensemble of random unitaries \cite{peterdamian,mezher}. 

Here, we find examples of small graph states along with measurement angles which generate ensembles of random unitaries which are partially invertible universal sets. 
By concatenating this seed construction in a specific way we generate ensembles with order $O(n^{4}t^{12})$ qubits which form an $\varepsilon$-approximate $t$-design on $U(2^{n})$. 

Translated into the circuit model, these MBQC circuits have a constant depth, since these circuits consist of non-adaptive measurements on a regular graph state, where each qubit is entangled with at most a constant number of neighbors \cite{juan}  \footnote{The total number of qubits - i.e the total number of horizontal lines in the quantum circuit \cite{nielsenchuang}- is $O(n^{4}t^{12})$.}. This observation could be very beneficial from the point of view of experimental implementation.

\subsection{Connection to quantum speedup}
There is currently a tremendous effort being made
to build a quantum computer, and develop quantum
technologies more generally. An important
benchmark for this ambitious project will
be proving a computational advantage over what
can be done with classical computers. Two results
in this direction have sparked a surge in research.
Boson sampling \cite{aaronson,aaronson2} and IQP \cite{iqp} are
subuniversal families of computation which can
be shown to be impossible to replicate efficiently
classically assuming some standard, and strongly
believed, complexity theoretic conjectures hold.
This is often referred to as \emph{quantum speedup}.
Since then, there have been many developments
of these and related models \cite{juan2,juan,gao,mann,bremner,bremnermontanaro,ising,aaronson,aaronson2,iqp,keshari,boxio}  
 to state a few. A review can be found in \cite{harrowmontanaro}. In all of these cases two features
are significant. Firstly, they do not require the
full capabilities of a universal quantum computer
and so are expected to be much simpler to implement,
and second they are all what is known
as sampling problems. That is, the statements
of difficulty are that a classical computer cannot
efficiently sample from the same distribution as
what can be achieved in these quantum architectures
efficiently.

More concretely, the statements run somewhat
as follows. Each of these computational models
is essentially a family of circuits followed by measurements,
the results of which follow a particular
distribution. If it is possible for a classical computer
to efficiently sample from this distribution,
then, certain strongly believed complexity conjectures
would be proved invalid. For proofs which
hold for approximate sampling, the standard conjectures
are of the form \cite{juan,juan2,gao,bremnermontanaro,bremner,mann,iqp} : \\
I) the polynomial hierarchy does not collapse to
the third level \cite{terhal}.\\
II) the average case of the associated problem
(usually $\sharp$P) is also hard ($\sharp$P).\\
III) the quantum circuit families considered output distributions
which are not too peaked - technically known as
anti-concentration \cite{boxio,juan,juan2,bremnermontanaro,bremner,mann}. 

One of the goals of the field now is to reduce
the number of required assumed conjectures, or
justify them, and understand their relationship
to other properties of a given architecture such
as universality. There are many architectures
demonstrating quantum speedup, suited to different
implementations with different versions of the
conjectures which link them in different ways to
different problems. The average case complexity
can be linked to conjectures of average case hardness
of solving certain Ising problems \cite{ising}, or of
Jones polynomials \cite{mann} for example. For several
architectures anticoncentration can be proven explicitly
\cite{juan2,mann}. The work of \cite{juan2,harrow2018,mann} shows an
interesting link between $t$-designs and anticoncentration. 

In this work, and as an application to our $t$-design graph gadgets, 
we use techniques from \cite{juan,juan2} and introduce new families of MBQC architectures showing a quantum speedup. We show that $every$ MBQC $2$-design constructed from partially invertible universal sets is hard to sample from classically, and yet we give an explicit example that can be prepared efficiently using $n$-qubit cluster states with $O(n^{3})$ columns - thereby presenting a quantum speedup. Because our architectures are $t$-designs by construction,  conjecture (III) is proven \cite{juan2,mann}, thus we only require 2 complexity theoretic conjectures in our proofs (namely, Conjectures I) and II) ). Also, because our gadgets have quite a regular structure, they can be translated into a constant depth  quantum circuit as explained above. This makes these architectures desirable for near-term experimental implementation. Finally, our architectures have a natural statement of verification; following from works on graph state verification \cite{juan,damianalexandra,gao,cramer,serfling}. 

\subsection{Families of universal ensembles}

In the final part of this work we explore how common universal ensembles are in the measurement based framework, and how they can be used for $t$-designs. We present two results in this direction, one analytical and the other numerical. 


Analytically, we show that $almost$ $any$ \footnote{Meaning that the set of angles which don't work form a set having zero Lebesgue measure \cite{lebesgue}.} assignment of fixed  $XY$ angle measurements on a $n=2^{\gamma}$-row, 2-column cluster state (where $\gamma$ is an integer) gives a random unitary set $\mathcal{U_{B}}$ which is universal in $U(2^{n})$. We use a Lie algebraic approach outlined by \cite{brylinski}, and observations in \cite{lloyd,barenco} to prove this result. In particular, when $\gamma=1$ we get that almost any assignment of fixed $XY$ angle measurements generates universal sets $\mathcal{U_{B}} \subset U(4)$, which in general can be invertible, partially invertible or non invertible. 

We then provide numerical evidence that for almost any fixed assignment of $XY$ measurements, the subdominant eigenvalue of the operator $M_t[\mu]=\frac{1}{|\mathcal{U_B}|}\sum_{i=1,...\mathcal{|U_B|}}U_{i}^{\otimes t} \otimes U_{i}^{* \otimes t}$  scales efficiently with $t$. \footnote{$U_{i} \in \mathcal{U_{B}}$. $M_t[\mu]$ is usually called the moment superoperator.}  If the numerical result is true, then together with the analytical result on universality, one can show from our techniques developed for the partially invertible case, that cluster state gadgets with almost any fixed $XY$ angle assignment give an efficient $n$-qubit $t$-design under concatenation. Further, the results imply that these gadgets are also hard to sample from classically under concatenation, and thus these gadgets may also be used  as architectures presenting a quantum speedup.

\bigskip

The organization of this paper is as follows. In Section 2 we define some preliminary notions for unitary $t$-designs \cite{dankert}, MBQC \cite{briegelrauss} and Classical simulability \cite{nest,iqp}. In Section 3 we present our main results.  
This begins in Theorem 1 where we present our main result on $t$-designs for partially invertible universal sets of unitaries. The following two corollaries apply this theorem to give general constructions in the circuit model and in the measurement based model given a partially invertible universal set.
We then give an explicit graph state construction generating a partially invertible universal set, which can be used to give explicit graphs with fixed angles and measurements implementing efficient approximate $t$-designs.
This is followed by Theorem~2 where we see that these constructions can be used to generate sampling problems which are provably hard to simulate for a classical computer, assuming standard complexity conjectures.
Finally in Theorem~3 we state our result for the universality of ensembles where measurements angles are fixed, with almost any fixed angles working. We present several examples and numerics suggesting that these also provide efficient $t$-designs.
Section 4 is devoted to the technical proofs. Finally, we draw out conclusions in Section 5.

\section{Preliminary notions}
\subsection{$\varepsilon$-approximate unitary $t$-designs and related concepts}
We start by defining what we mean by a random unitary ensemble.
\begin{definition}
\label{df1}
A random unitary ensemble in $U(N)$  is a couple  $\{p_{i},U_{i} \in \mathcal{U}\}_{i=1,...|\mathcal{U}|}$ (or simply $\{p_{i},U_{i}\}$ for ease of notation), where each unitary  $U_{i} \in \mathcal{U}$ is drawn with probability $p_{i} \geq 0$, and $\sum_{i}p_{i}=1$, with $\mathcal{U} \subset U(N)$.
\end{definition}
We now formalize the definition of $\varepsilon$-approximate unitary $t$-designs (or just $\varepsilon$-approximate $t$-designs for simplicity) which are the main objects of study in this work.
\begin{definition}
\label{df2}
\cite{brandao,nakata} Let $\mathcal{H}$ be the $n$-qubit Hilbert space $(\mathbb{C}^{2})^{\otimes n}$. A random unitary ensemble $\{p_{i},U_{i}\}$ with $U_{i} \in U(2^{n})$ is said to be an $\varepsilon$-approximate $t$-design  if the following holds 

\begin{equation}
\label{eq1}
(1-\varepsilon)\int_{U(2^{n})}U^{\otimes t}\rho U^{\dagger \otimes t}\mu_H(dU) \leq \sum_{i} p_{i} U_{i}^{\otimes t}\rho U_{i}^{\dagger \otimes t} \leq \\ (1+\varepsilon)\int_{U(2^{n})}U^{\otimes t}\rho U^{\dagger \otimes t}\mu_H(dU)
\end{equation}
for all $\rho \in B(\mathcal{H} ^{\otimes t})$, where $\mu_{H}$ denotes the Haar measure on $U(2^n)$. For positive semi-definite matrices $A$ and $B$, $B \leq A$ means $A-B$ is positive semi-definite, $\varepsilon$ and $t$ are positive reals.
\end{definition}
Definition \ref{df2} is sometimes referred to as the $strong$ $definition$ of a $\varepsilon$-approximate $t$-design \cite{brandao,harrow2018}. Note that when $\varepsilon=0$ one recovers the case of exact $t$-designs \cite{dankert,zhu}. Similarly, one can define an approximate $t$-design in terms of various norms, depending on the application in mind \cite{brandao,harrow2018}.
 
 To prove our results, we will study the properties of an operator referred to as the \emph{moment superoperator} $M_t[\mu]$ defined as follows \cite{wintonviola, harrow2des,brandao,mezher}.
\begin{definition}
\label{defmomentsuperop}
For a random unitary ensemble $\{p_{i},U_{i} \in \mathcal{U}\}$, 
 \begin{equation}
\label{eq2}
M_t[\mu]=\sum_{i}p_{i}U_{i}^{\otimes t,t},
\end{equation}
where $\mu$ is the probability measure \footnote{ As shown in \cite{harrow2des} one can shift between a probability distribution over a discrete ensemble $\{p_{i},U_{i}\}$ and a continuous distribution by defining the measure $\mu=\sum_{i}p_{i}\delta_{U_{i}}$.}  over the set $\mathcal{U}$ which results in choosing $U_{i} \in \mathcal{U}$ with probability $p_{i}$, and
 $U^{\otimes t,t}=U^{\otimes t} \otimes U^{* \otimes t}$, and $U^{*}$ is the complex conjugate of $U$.\\
\end{definition}
A useful concept we will frequently make use of is that of an $(\eta,t)$-tensor product expander \cite{hastings,hastings2} (TPE) defined as follows.
\begin{definition}\cite{hastings,hastings2}
\label{deftpe}
A random unitary ensemble $\{p_{i},U_i\}$ is said to be an $(\eta,t)$-TPE if the following holds,
\begin{equation}
\label{eq3}
||M_{t}[\mu]-M_{t}[\mu_{H}]||_{\infty} \leq \eta < 1,
\end{equation}
where $M_t[\mu_H]=\int_{U(2^n)}U^{\otimes t,t} \mu_H(dU)$.
\end{definition}

In particular, we will adopt the usual path \cite{brandao,harrow2des,nakata,mezher,peterdamian} of showing that 
our random unitary ensembles are $(\eta,t)$-TPE's, then use the following proposition to obtain statements about $t$-designs. 
\begin{proposition} \cite{brandao,harrow2des,nakata}
\label{prop1}
If $\{p_{i},U_{i} \in \mathcal{U}\}$ is an $(\eta,t)$-TPE, then the \textbf{k-fold concatenation} of $\{p_{i},U_{i} \}$:  
$\{\prod_{j=1,...k}p_{\pi(j)},\prod_{j=1,...k}U_{\pi(j)}\}$ \footnote{ Note that the random ensemble $\{\prod_{j=1,...k}p_{\pi(j)},\prod_{j=1,...k}U_{\pi(j)}\}$  has a moment super operator $M_t[\mu_k]=(M_t[\mu])^{k}$ \cite{wintonviola}.} is an $\varepsilon$-approximate $t$-design when 
\begin{equation}
\label{eq4}
k \geq \frac{1}{log_{2}(\frac{1}{\eta})}(4nt+log_{2}(\frac{1}{\varepsilon})).
\end{equation}
 Here $\pi$ is a  function acting on $\{1,...,k\}$, resulting in a set $\{\pi(1),...\pi(k)\}$ where $\pi(j) \in \{1,...,|\mathcal{U}| \}$, the $\pi(j)'s$ can be identical. There are $|\mathcal{U}|^k$ such functions $\pi$ and the $k$-fold concatenation includes all of them. \\
\end{proposition}

\begin{proof}\cite{brandao,nakata}
$||\delta_{\mu_{k}} -\delta_{\mu_{H}}||_\diamond  \leq 2^{2nt} \eta^k$. 
Where $||.||_\diamond$ is the diamond norm \cite{brandao}, and $\delta_\mu$ is defined as $\delta_{\mu}(X)=\int_{U \sim \mu} U^{\otimes t} X U^{\dagger \otimes t} d\mu(U)$. Furthermore, 
if $||\delta_{\mu_{k}} -\delta_{\mu_{H}}||_\diamond \leq \frac{\varepsilon}{ 2^{2nt} }$,
 then $\{\prod_{j=1,...k}p_{\pi (j)},\prod_{j=1,...k}U_{\pi(j)}\}$ is an $\varepsilon$ approximate $t$-design in the strong sense ($cf $. Definition \ref{df2}) \cite{brandao}. The value of $k$ in Proposition \ref{prop1} is thus obtained by setting: $\frac{\varepsilon}{2^{2nt}} \geq 2^{2nt}\eta^{k}$ .
\end{proof}
We will make use of the following fact proven in \cite{harrow2des}.
\begin{proposition}\cite{harrow2des}
\label{prp2}
If $\mu$ is a probability measure with support on a universal gate set $\mathcal{U} \subset U(2^n)$ \footnote{ In other words, for all $U \in \mathcal{U}$, $\mu(U) \neq 0$. }, then the following inequality holds for any positive integer $t$ .
\begin{equation}
\label{eq5}
||M_t[\mu]-M_t[\mu_H]||_\infty < 1.
\end{equation}
\end{proposition}

In recent work $\varepsilon$-approximate $t$-designs have been shown to anti-concentrate \cite{juan2,mann}. Fundamentally, anti-concentration is a statement about probability distributions. For circuits that anti-concentrate, the probability of occurrence  of  $most$  outcomes  is  reasonably large \cite{vazirani}. The property of anti-concentration, as mentioned in the introduction, plays a key role in proofs of hardness approximate classical sampling \cite{juan,juan2,bremner,vazirani}. 
 We now present a theorem  on the anti-concentration of $t$ $\geq$ 2 -designs, shown  in  \cite{juan2} ( a similar result was derived independently in \cite{mann}).

\begin{proposition} \cite{juan2}
 Let \{$p_{i}$,$U_{i}$\} be an $\varepsilon$-approximate 2-design on U($2^{n}$). Let $\ket 0 ^{\otimes n}$:=$\ket 0 $ be an $n$-qubit input state to which we apply a unitary $U_{i}$ from the 2-design. Then, for any $ x\in \{0,1\}^{n}$ there exists a universal constant $0 \leq \alpha \leq 1$ such that:
\begin{equation}
\label{eq anticonc}
Pr_{U_{i} \sim \mu} (|\bra x U_{i} \ket 0|^{2} > \frac{\alpha(1-\varepsilon)}{2^{n}} ) \geq \frac{(1-\alpha)^{2} (1-\varepsilon)}{2(1+\varepsilon)}.
\end{equation}
$\mu$ being the probability measure over the $t$-design that results in choosing $U_{i}$ with probability $p_{i}$.
\end{proposition}

\subsection{Measurement Based Quantum Computation (MBQC)}

As mentioned, MBQC is a natural landscape for the generation of random unitary ensembles. This section shows how one can generate such ensembles in the language of MBQC. We begin by defining graph states (see e.g. \cite{graphstate}), a main component of MBQC.
\begin{definition}
A graph state $|G\rangle$ is a pure entangled multipartite state of $n$ qubits in one-to-one correspondance with a graph $G=\{E,V\}$ of $n$ vertices. Every vertex $i \in V$ represents a qubit, and each edge $(i,j) \in E $ can be understood as a preparation entanglement. 
\begin{equation}
\label{eq6}
|G\rangle=\prod_{{(i,j)}\in E}CZ_{i,j}|+\rangle^{\otimes n},
\end{equation}
where $|+\rangle=\frac{1}{\sqrt{2}}(|0\rangle+|1\rangle)$ and $CZ_{i,j}$ is the controlled $Z$ gate applied to qubits $i$ and $j$ (see e.g. \cite{nielsenchuang}).
\end{definition}
For the purposes of computation, a subset of qubits $I \subset V$ is defined as the computational input with initial input state $|\psi_{in}\rangle_{I}$, and the associated \emph{open} graph state has the form 
\begin{equation}
\label{eq7}
|G(\psi)\rangle=\prod_{{(i,j)}\in E}CZ_{i,j}|\psi_{in}\rangle_{I} \otimes |+\rangle_{V/I}.
\end{equation}
A cluster state \cite{raussbriepra} is a particular type of graph state whose corresponding graph is a regular two dimensional square lattice.  In MBQC, computation is carried out by performing measurements on all but a subset $O \subset V$ of  qubits. In general one has $|O|\geq |I|$, though here we are concered only with the case $|I|=|O|$. 
By performing the measurements adaptively on a universal resource state (such as the cluster state) - via some corrective strategy such as the gflow \cite{elham} - one can implement any desired unitary $U \in U(2^{|O|})$ on the input state, which is teleported to the (unmeasured) output position by the end of the computation. At the end of the computation, we are left with the following state

\begin{equation}
\label{eq8}
|OUT\rangle=|M\rangle_{V/O}\otimes U|\psi_{in}\rangle_{O},
\end{equation}
where $|M\rangle_{V/O}$ represents the measurement outcomes, performed adaptively. Cluster states are universal resources for measurement based quantum computation (MBQC) \cite{briegelrauss,raussbriepra}, even when all the measurement angles are chosen from the XY plane \cite{mantri}.

On the other hand, performing the measurements non-adaptively (that is, simultaneously and without a corrective strategy) generates a random unitary ensemble $\{p_{i},U_{i}\}$, seen from noting that we can (for an appropriate choice of measurement bases) write $|G(\psi)\rangle$ as, 
\begin{equation}
\label{eq9}
|G(\psi)\rangle=\sum_{i} \sqrt{p_{i}}|M_{i}\rangle_{V/O}\otimes U_{i}|\psi_{in}\rangle_{O}.
\end{equation}
$|M_{i}\rangle_{V/O}$ denotes a possible string of measurement results which implements unitary $U_{i}$ on the input state. This measurement result occurs with probability $p_{i}$. 
In our case -  MBQC on unweighted cluster states - the probability distribution is  uniform,  $p_i=\frac{1}{2^{|V/O|}}$. Figure \ref{fig1} shows an example of a non-adaptive MBQC scheme on a 2-row, 2-column cluster state at XY plane measurements $\alpha$, $\beta$, $|V|=4$, and $|O|$=2. This non-adaptive scheme generates the random unitary ensemble,
\begin{equation}
\label{eq10}
\{\frac{1}{4},CZ(HZ(\alpha+m_{1}\pi)\otimes HZ(\beta+m_{2}\pi))\}
\end{equation}
where $H$ is the Hadamard gate, $Z(\alpha)=e^{-i\frac{\alpha}{2}Z}$ is a rotation by angle $\alpha$ around the Z axis, $CZ$ is the controlled-Z gate and $m_i\in\{0,1\}$ represents the measurement outcome of qubit $i$, following the convention that $m_1=0$  is taken to mean measurement outcome corresponding to a projection onto $|+_{\alpha}\rangle=\frac{|0\rangle + e^{i\alpha}|1\rangle}{\sqrt{2}}$ 
(respectively  $|+_{\beta}\rangle=\frac{|0\rangle + e^{i\beta}|1\rangle}{\sqrt{2}}$  for $m_{2}=0$) and $m_{1}=1$ a projection onto $|-_{\alpha}\rangle=\frac{|0\rangle - e^{i\alpha}|1\rangle}{\sqrt{2}}$ (respectively  $|-_{\beta}\rangle=\frac{|0\rangle - e^{i\beta}|1\rangle}{\sqrt{2}}$  for $m_{2}=1$).

 \begin{figure}[H]
\begin{center}
\graphicspath{}
\includegraphics[trim={2 10cm 250 0cm} , scale=0.3]{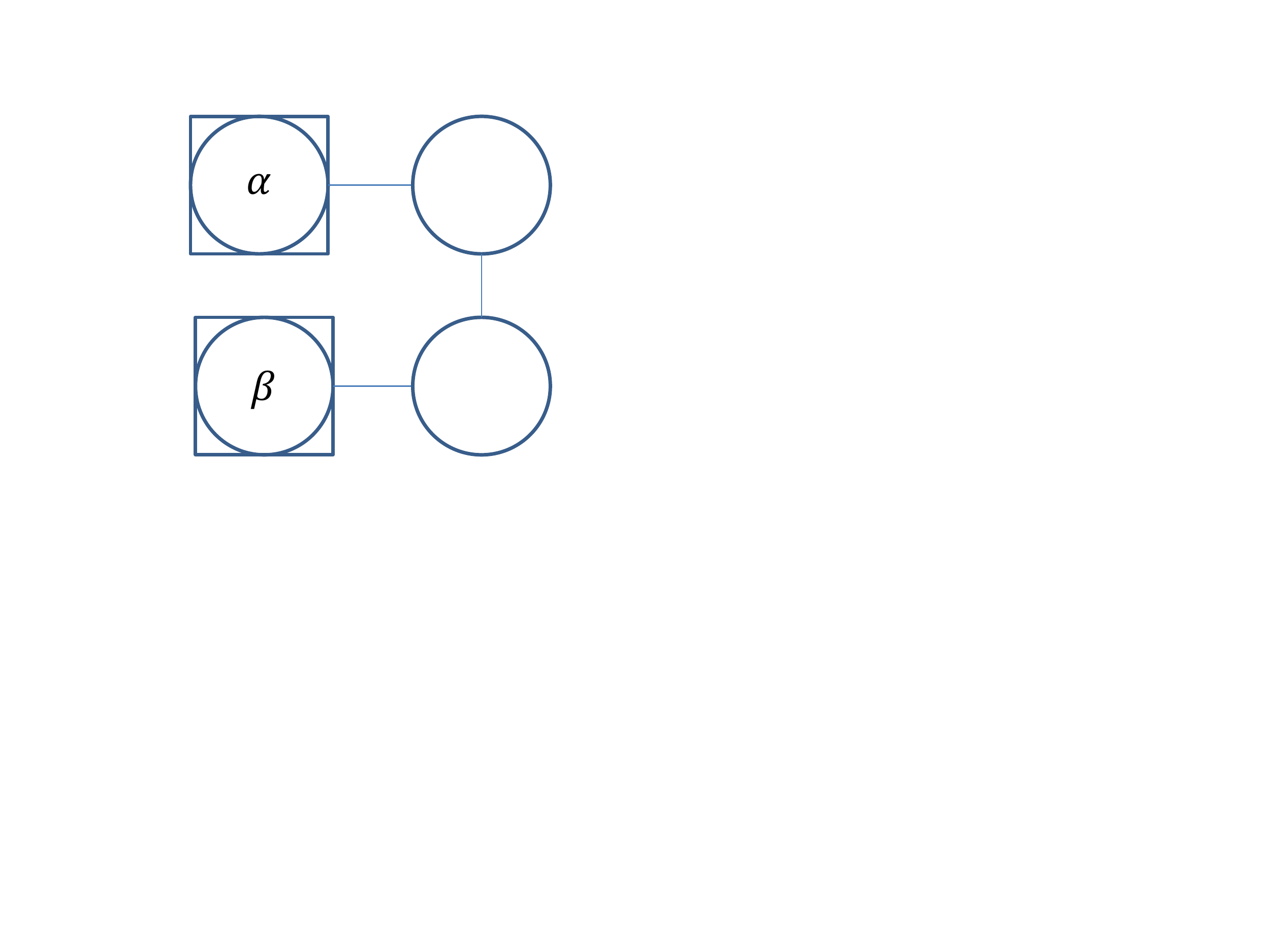}
\caption{MB scheme on a 2-row, 2-column cluster state. The input qubits (squared circles) when measured non-adaptively at XY angles $\alpha$ and $\beta$ apply to the unmeasured output qubits (empty circles) a random unitary of the ensemble of Equation (\ref{eq10}) chosen with a uniform probability of $\frac{1}{4}$. The horizontal and vertical lines are preparation entanglements. }
\label{fig1}
\end{center}
\end{figure}

\subsection{Notions of simulability, and structure of a standard hardness of approximate classical sampling
proof} 

Let $\{C_n\}$ be a family of quantum circuits with $n$ input qubits. Suppose also that this family satisfies some uniformity condition
(e.g. \cite{iqp,hoban}) to ensure no computationally unreasonable preparations are required with varying inputs
of the family. Let $P_n$ denote the probability distribution associated with measuring the outputs
of $C_n$ in the computational (Z) basis . We say that the circuit family $\{C_n\}$ is \textit{classically simulable
in the strong sense} if any output probability in $P_n$, and any marginal probability of $P_n$
can be approximated up to $m$ digits of precision by a classical $poly(n,m)$ time algorithm \cite{iqp}. 

Because the output probabilities of universal-under-post-selection quantum circuits are $\sharp$Phard to exactly compute in worst-case \cite{nest} (and even $\sharp$P-hard to approximate up to relative error 1/4 in worst-case \cite{ising,goldberg} ), this makes the
task of strong classical simulability formidable even for quantum computers. In order to find tasks where one clearly sees a quantum advantage
(in other words, tasks which are hard for classical computers but which can nevertheless be performed efficiently on some, possibly nonuniversal
\cite{iqp}, quantum device), one needs to introduce the notion of classical simulability in the weak sense. 

\textit{Classical simulability in the
weak sense} means that the classical algorithm can
sample, i.e output $x$ (one of the possible outputs
of circuit $C_n$) with probability $p_x \in P_n$, in
$poly(n)$ time. For practical purposes (due to experimental imperfections), one usually requires a notion of approximate classical simulability in the weak sense (henceforth referred to as \textbf{approximate classical
sampling}), of which many exist \cite{nest,iqp}.
In our work we adopt the following definition of
approximate classical sampling (taken from \cite{iqp}).
\begin{definition}
We say that a family of circuits $\{C_n\}$ on $n$-qubtis
where each $C_n$ has a set of possible outputs $x$
with an associated output probability $p_x$ is
\emph{approximately classically simulable in the weak sense} (i.e admits an approximate classical sampling), up to an
$l_1$-norm distance (or total variation distance ) $\sigma$,
if there exists a $poly(n)$ time classical algorithm
A sampling $x$ with probability $p_{A_x}$ for which the
following holds 
\begin{equation}
\label{eq class sim}
\forall C_n \in \{C_n\} ,\sum_{x} |p_{x}-p_{A_{x}}| \leq \sigma . 
\end{equation}
\end{definition}
The expression of quantum speedup is precisely
that no classical $poly(n)$ time algorithm A
exists which can approximately sample (in the
sense of Equation (\ref{eq class sim}) ) given that some complexity
theoretic conjectures hold.

The argument for quantum speedup comes
from two directions. 
Firstly consider the power
of a classical algorithm which is able to approximately
sample from the distribution $p_x$ as defined
above. The trick is to boost this up from sampling
$p_x$, to approximating $p_x$ (that is a simulation in
the strong sense). This is the role of Stockmeyer’s
counting theorem, and it does this at the third
level of the polynomial hierarchy (PH) \cite{stockmeyer}. In
particular it says that there is an algorithm at
the third level (concretely in $FBPP^{NP}$) which
takes the classical algorithm for sampling $p_x$ and outputs
an approximation of $p_x$, up to \textbf{additive error}. The remaining steps
on the classical side are to make this approximation
stronger, and work for \textbf{relative errors}, which
is what one wants for realistic experimental errors \cite{bremner,juan,juan2}.
To do this step we rely on the fact that the
output distributions of our families of circuits are
not too peaked, a property known as anti-concentration \cite{harrowmontanaro}.
This is where conjecture III) (see the introduction) comes in. The final statement
is that for a fraction $f$ of the family of circuits
considered, the output distrubution $p_x$ can
be approximated up to a relative error (see Section 5).

The other direction comes from the known
hardness of sampling quantum distributions. The
first statement in this direction says that appoximating
$p_x$ (exactly, or up to relative error) is $\sharp$P
hard in the worst case (that is, for one or more
of the circuits in the family), as mentioned earlier.
This is standard following universality of the
circuit families \cite{bremner,bremnermontanaro,gao,juan,juan2,mann,ising,boxio,harrowmontanaro}. The
difficulty here is to match this to the statement
about the fractions of the circuits considered, in
order to match the relative error approximation
we would have classically from above. To this
end we are forced to add an assumption about
the hardness of the average case (over the circuit
family). This is the content of conjecture II) (see
introduction) and there are various justifications
for this, depending on which families
of problems it is related to \cite{juan,juan2,bremnermontanaro,ising,bremner,mann,boxio}.
Bringing these together we have, assuming conjectures
II) and III), that the existence
of a classical algorithm approximately sampling
$p_x$ (in the sense of Equation (\ref{eq class sim}) ) implies that solving
a $\sharp$P hard problem can be achieved at the
third level of the PH. This implies the collapse of
the PH to its 3rd level by a theorem of Toda \cite{toda}.
Thus, if one believes this cannot be possible (conjecture
I) in introduction) one is forced to give up the
possibility of such a classical sampling algorithm.

\section{Main Results}
Let $\mathcal{U_B} \subset U(4)$ be $any$ partially invertible universal set in $U(4)$. Let  $\mathcal{U_M} \subset \mathcal{U_B}$, with $\mathcal{U_M}$ containing unitaries and their inverses and with unitaries composed of algebraic entries, and its complement $\mathcal{U_{B/M}} \subset \mathcal{U_B}$ such  that $\mathcal{U_M}$ and $\mathcal{U_{B/M}}$ are both universal in $U(4)$. Define
\begin{equation}
\label{eqpartinv}
B=\{\frac{1}{|\mathcal{U_B}|},U_{i} \in \mathcal{U_B}\}.
\end{equation}  
Denote the $k$-fold concatenation of $B$  by
\begin{equation}
\label{eqbk}
    B^k=\{\frac{1}{|\mathcal{U}_{\mathcal{B}^k}|},\prod_{j=1,...k}U_{\pi(j)} \in \mathcal{U}_{\mathcal{B}^k} \}
\end{equation}
 where $U_{\pi(j)} \in \mathcal{U_B}$, and $\pi$ is a function defined as in Proposition \ref{prop1}. 
Define \footnote{ This definition of $block(B^k)$ is for even $n$ , the odd $n$ case follows straightforwardly.}  
\begin{equation}
\label{eqblockbk}
 block(B^k)=\{\frac{1}{|\mathcal{U}_{\mathcal{B}^k}|^{n-1}},(1_{2 \times 2} \otimes U^{j_1}_{2,3} \otimes U^{j_2}_{4,5}\otimes ...\otimes U^{j_{\frac{n}{2}-1}}_{n-2,n-1} \otimes 1_{2 \times 2} ). (U^{j_\frac{n}{2}}_{1,2}\otimes U^{j_{\frac{n}{2}+1}}_{3,4}\otimes ...\otimes U^{j_{n-1}}_{n-1,n})\in \mathcal{U}_{block(\mathcal{B}^{k})} \}, 
 \end{equation}
 where $U^{j}_{i,i+1}  \in \mathcal{U}_{\mathcal{B}^k}$, $i \in \{1,...,n-1\}$ and $j \in \{1,...,|\mathcal{U}_{\mathcal{B}^k}|\}$. Let $block^{L}(B^k)$  be the $L$-fold concatenation of $block(B^k)$, defined as
 \begin{equation}
 \label{eqblockblk}
 block^{L}(B^k)=\{\frac{1}{|\mathcal{U}_{\mathcal{B}^k}|^{(n-1)L}} ,\prod_{j=1,...,L}U_{\pi(j)} \in \mathcal{U}_{block^L(\mathcal{B}^{k})}\} 
 \end{equation}
 where here also $\pi$ is defined as in Proposition \ref{prop1}, and $U_{\pi(j)} \in \mathcal{U}_{block(\mathcal{B}^{k})}$  .

 Finally, let $a=\frac{|\mathcal{U_M}|}{|\mathcal{U_B}|}$.  Our first main result is the following theorem which holds for the above defined partially invertible universal set $\mathcal{U_{B}}$:
\begin{theorem}
\label{th1}
For any  $0<\varepsilon _{d} <1$, and for some $0<C<1$,  if : 
\begin{equation}
k \geq \frac{1}{log_{2}(\frac{1}{1+(C-1)a})}(10t+n^2t-nt+n+log_{2}(\frac{1}{\varepsilon^{'}}))
\end{equation}
and
\begin{equation}
L \geq \frac{1}{log_2(\frac{1}{\varepsilon^{'}+P(t)})}(4nt +log_2(\frac{1}{\varepsilon_d})), 
\end{equation}
where 
\begin{equation}
\label{eqpt}
P(t)=(1+\frac{(425\lfloor{log_2(4t)} \rfloor^2 t^5 t^{3.1/log(2)})^{-1}}{2})^{-1/3}, 
\end{equation}
$\varepsilon^{'} <1-P(t),$ and $n \geq \lfloor{2.5log_2(4t)} \rfloor$, then $block^{L}(B^{k})$, formed from partially invertible universal set $\mathcal{U_{B}}$, is a $\varepsilon_{d}-$ approximate $t$-design on $U(2^{n})$, for any $t$. 
\end{theorem}
Here $\lfloor.\rfloor$ denotes the floor function. An Immediate corollary to the above theorem is the following less technical statement.
 
\begin{corollary}
\label{cor1}
Let $B$ be the random unitary ensemble formed by  chosing uniformly at random from a partially invertible universal set . Random quantum circuits on $n$-input qubits of depth $D=2.k.L=O(n^{3}t^{12})$ \footnote{ Note that, as in \cite{brandao}, $\dfrac{1}{log_{2}(\dfrac{1}{P(t)+\varepsilon^{'}})} \sim O(t^{9.47}log^{2}(t)) < O(t^{10})$, as $t \to \infty $ and thus $k.L \sim O(t^{10}).O(n^3t^2)=O(n^3t^{12})$.}  
and described as follows (for $n$ even, odd $n$ case follows straightforwardly.) \begin{enumerate}
    \item For steps 1 to $k$ (layer $j=1$), apply  unitaries of the form $U_{1,2} \otimes U_{3,4}...\otimes U_{n-1,n}$, where the $U_{i,i+1}$'s are random unitaries sampled independently from the random unitary ensemble $B$, and acting non-trivially on input qubits $i$ and $i$+1.
    \item For steps $k$ to $2k$ (layer $j=2$),  apply  unitaries of the form $U_{2,3} \otimes U_{4,5}...\otimes U_{n-2,n-1}$, where the $U_{i,i+1}$'s are random unitaries sampled independently from the random unitary ensemble $B$, and acting non-trivially on input qubits $i$ and $i$+1.
    \item Repeat  1. for every odd numbered layer $j$ formed of $k$ steps, and repeat 2. for every even numbered layer $j$ formed of $k$ steps, for $j=3,...,2L$.
     
\end{enumerate}
are $\varepsilon_d$-approximate $t$-designs, for any $t$ and for $n \geq \lfloor{2.5log_2(4t)} \rfloor$ .
\end{corollary}

As shown in \cite{mezher}, one can generate random ensembles in  MBQC by connecting 2-qubit graph gadgets in a regular way. 
Given a graph gadget $G_B$, which gives an ensemble over a partially invertible universal set, we will see that Figures \ref{fig2a}, \ref{fig2b} and \ref{figlblock} show how to compose copies of $G_B$ to get the  $n$-qubit cluster state gadget $LG_{block(B^{k})}$ giving rise to the ensemble $block^L(B^{k})$.

We give explicit examples of such gadgets $G_B$ below (see Fig.~\ref{fig3}).
Obtaining the $k$-fold concatenation $B^k$ of the random unitary ensemble $B$  translates in MBQC to constructing a graph state gadget $kG_B$ which is formed of sticking together $k$ copies of $G_B$.  More precisely, if $G_B$ is a cluster state gadget formed of $m$ columns and 2-rows, then $kG_B$ is a cluster state gadget formed of $k(m-1)+1$ columns and 2-rows, where the measurement angles are repeated after each block of $m$ rows, see Figure~\ref{fig2a}.
Then, connecting these $kG_B$ gadgets in a brickwork like fashion gives rise to the $block(B^{k})$. We call this the graph state gadget $G_{block(B^{k})}$ and it is represented in Figure~\ref{fig2b}.
Finally, taking $L$ copies of these, concatenated after each other as in Figure \ref{figlblock} gives rise to a $t$-design, as is captured in the following corollary - which is a direct consequence of Theorem \ref{th1}, and the graph state translation to MBQC.
\begin{corollary}
\label{cor2}
If $G_B$ is a 2-qubit graph state gadget giving rise to a random unitary ensemble $B$ over a partially invertible universal set $\mathcal{U_B}$, then, for any $0<\varepsilon _{d} <1$,  and for some $0<C<1$, the graph state gadget $LG_{block(B^k)}$ applies to its unmeasured $n$ output qubits a unitary sampled from a $\varepsilon_{d}$-approximate $t$-design on $U(2^n)$ when,

$$k \geq \frac{1}{log_2(\frac{1}{1+(C-1)a})}(8t+(nt+2t+n^{2}t-2nt+n)+ log_2(\frac{1}{\varepsilon^{'}})), $$ $$ L \geq \frac{1}{log_2(\frac{1}{\varepsilon^{'}+P(t)})}(4nt +log_2(\frac{1}{\varepsilon_d})), $$ $\varepsilon^{'}< 1-P(t)$ , and $n \geq \lfloor{2.5log_2(4t)} \rfloor$,  for any t. \footnote{ A particular choice of $\varepsilon^{'}$ can be $\varepsilon^{'}=a(1-P(t))$, where $0< a < 1$ is a constant independent of $t$.}  \\
\end{corollary}

The graph state gadget $G_B$ in Figure \ref{fig3} generates a random unitary ensemble where elements of a partially invertible universal set are selected uniformly at random. This is proven in the appendix.

\bigskip
\begin{figure}[H]
\begin{center}
\graphicspath{}
\includegraphics[trim={10 4cm 0 0cm} , scale=0.2]{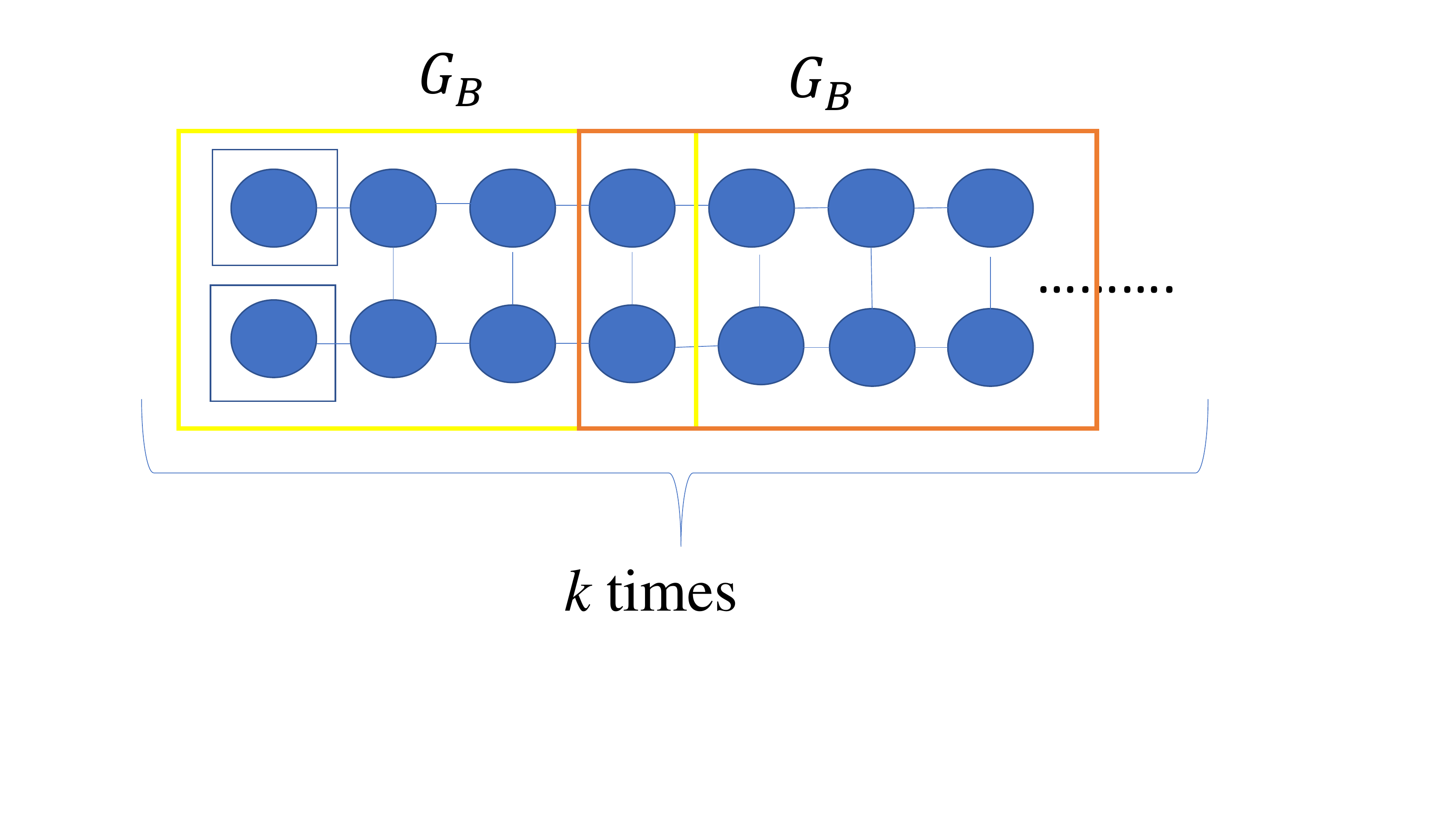}
\caption{Graph state gadget $kG_B$ giving rise to the random ensemble $B^k$.}
\label{fig2a}
\end{center}
\end{figure}

\begin{figure}[H]
\begin{center}
\graphicspath{}
\includegraphics[trim={2 0cm 150 0cm} , scale=0.3]{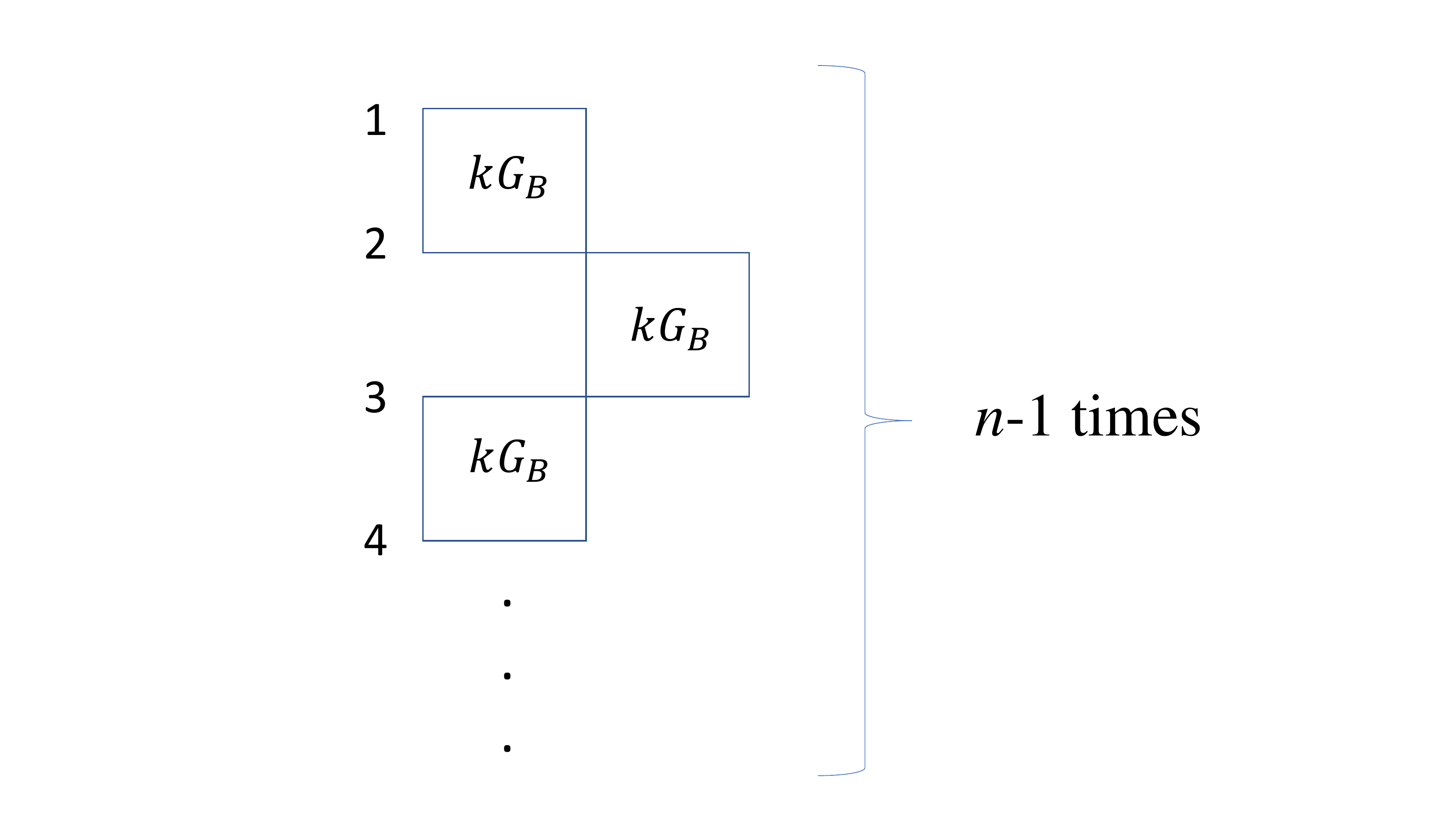}
\caption{Graph state gadget $G_{block(B^k)}$ giving rise to the random ensemble $block(B^k)$. The squares are 2-qubit gadgets $kG_B$.The empty 3 sided square means that there is no vertical entanglement.}
\label{fig2b}
\end{center}
\end{figure}
\begin{figure}[H]
\begin{center}
\graphicspath{}
\includegraphics[trim={2 0cm 200 0cm} , scale=0.3]{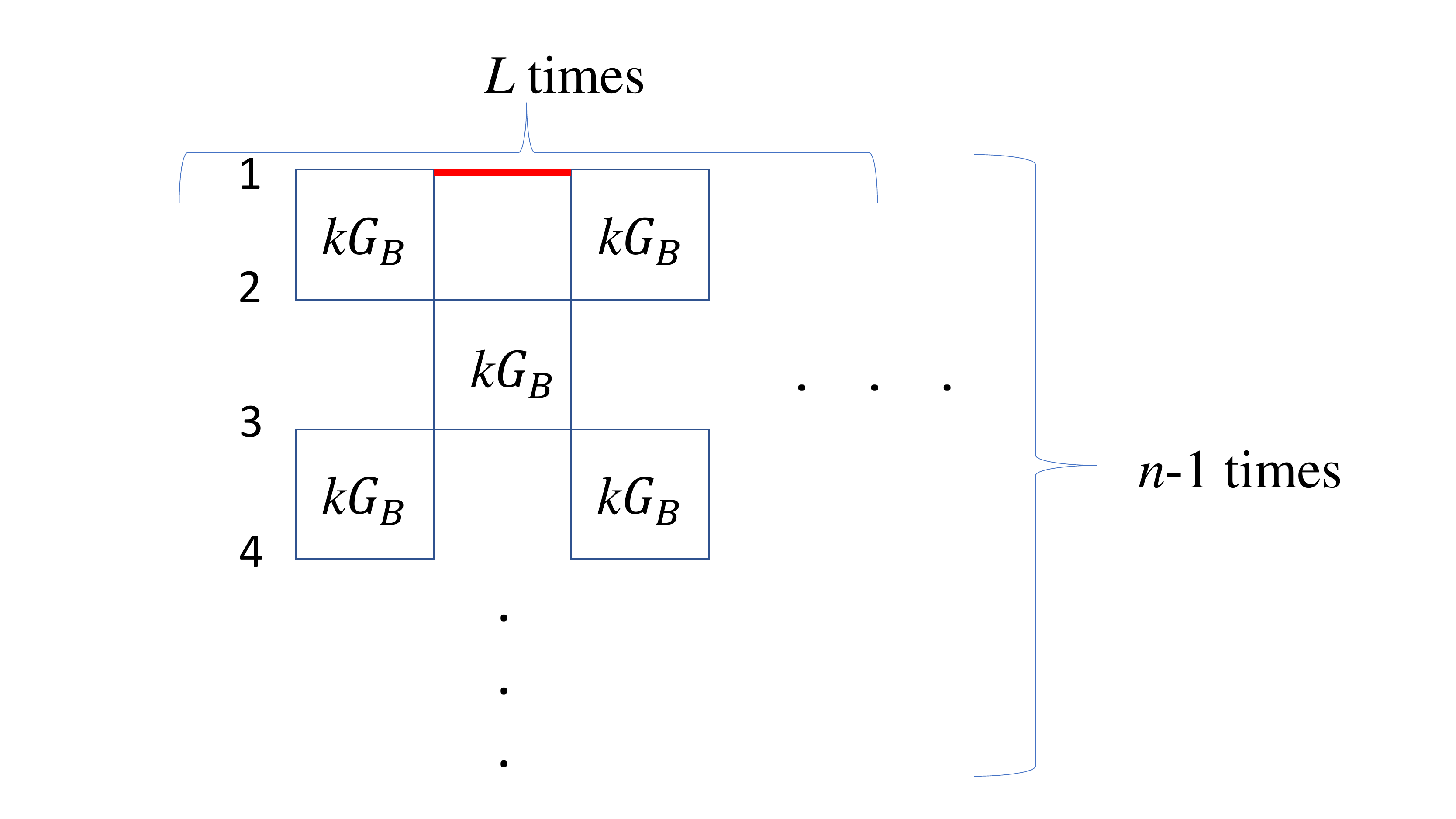}
\caption{ Graph state gadget $G_E:=LG_{block(B^k)}$, giving rise to the ensemble $E=block^{L}(B^k)$.
The horizontal red line is a preparation entanglement.}
\label{figlblock}
\end{center}
\end{figure}

\begin{figure}[H]
\begin{center}
\graphicspath{}
\includegraphics[trim={2 6cm 0 0cm} , scale=0.4]{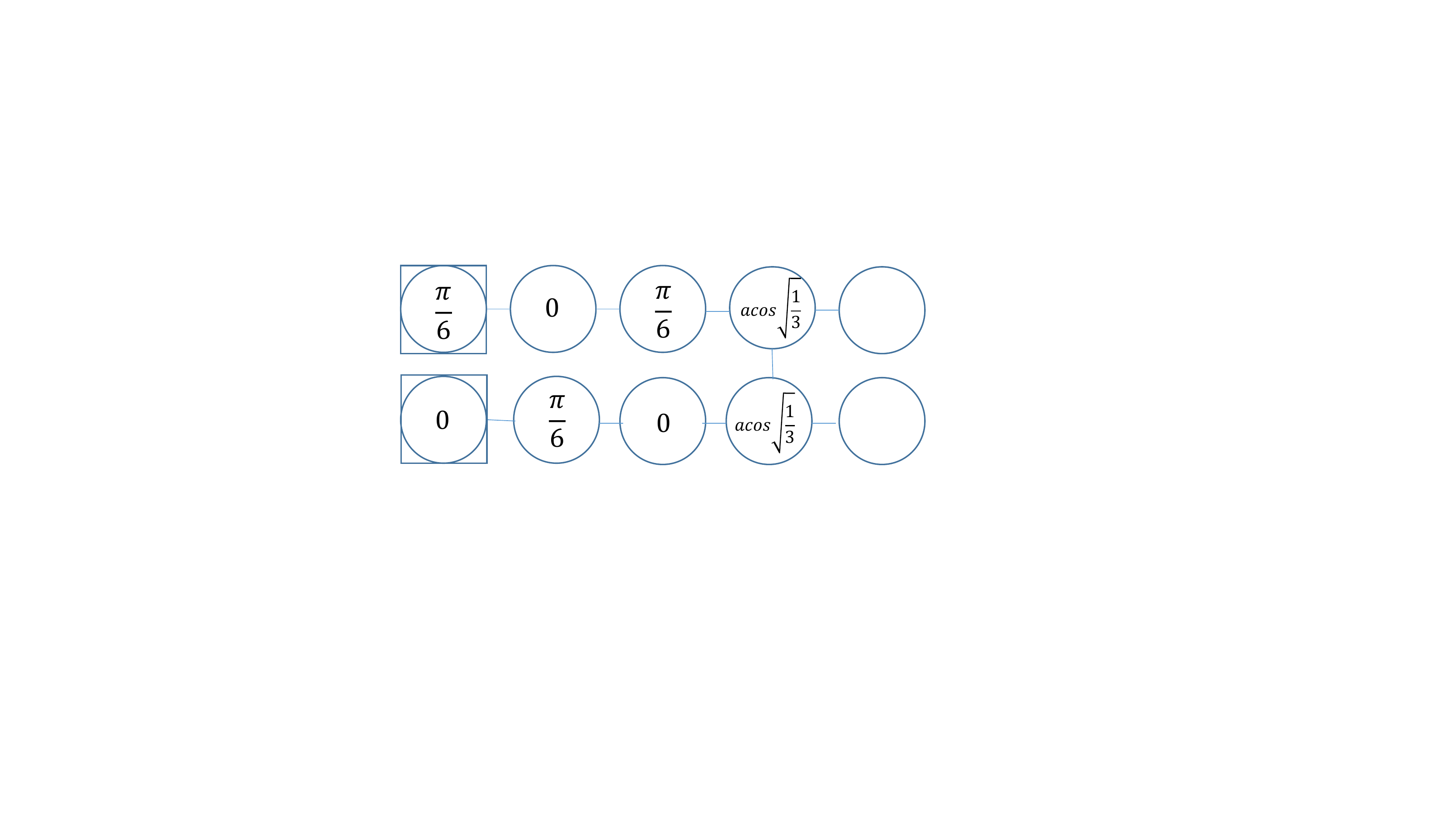}
\caption{ Graph state gadget $G_B$ giving rise to a partially invertible universal set.}
\label{fig3}
\end{center}
\end{figure}
Our next main result concerns sampling problems and quantum speedup using graph state gadgets $LG_{block(B^k)}$, Figure \ref{figlblock}.
For ease of notation, we denote $E=block^L(B^k)$, $\mathcal{U}_{block^L(\mathcal{B}^{k})}=\mathcal{U}_E$ and $G_E:=LG_{block(B^k)}$. 
Note that the total number of qubits of $G_E$ is $O(n.L.k)$, out of which $n$ qubits are identified input, and another $n$ qubits as output. 
The expressions of $L$ and $k$ are given in Theorem \ref{th1}. 
We will fix $\varepsilon_d$ to a specific value (which we will calculate in Section 5) and $t=2$, which gives $O(n.L.k)=O(n^4)$. 

Consider the sampling problem consisting of measuring the output qubits of $G_E$ in the computational basis, with the input state of $G_E$ being $|0\rangle^{\otimes n}:=|0\rangle$ and let $x$ be a bit string representing the outcomes of measurement of the output qubits of  $G_E$, and $y$ a bit string representing the outcomes of measurements performed on the non-output qubits. 
All measurements are non-adaptive, with angles defined by the graph state gadgets, and can be performed simultaneously. 
Let $$\ket{\psi}: =\prod_{i,j}CZ_{i,j}(\ket{+}^{\otimes O(n^4)-n} \otimes \ket{0}^{\otimes n}):=\prod_{i,j}CZ_{i,j}(\ket{+}^{\otimes O(n^4)-n} \otimes \ket{0})$$ denote the graph state corresponding to the graph state gadget $G_E$ before any  measurements are performed. 
This sampling gives rise to a probability distribution over $x \in \{0,1\}^{n}$ and $y \in \{0,1\}^{|V|-n}$, with $|V|=O(n^4)$ is the number of vertices in the graph state, defined by 
: 
\begin{equation}
\label{eq15}
D(x,y)=\{   p(x,y)=\\ |\langle{x,y} | \psi \rangle|^{2}=\frac{1}{2^{O(n^4)-n}} |\bra{x} U_{y} \ket{0}|^{2} \} ,
\end{equation}
where $U_{y} \in \mathcal{U}_E$, $|\mathcal{U}_E|=2^{O(n^4)-n}$,  and $\ket{x,y}=\ket{y} \otimes \ket{x}$. The relation $|\langle{x,y} | \psi \rangle|^{2}=\frac{1}{2^{O(n^4)-n}} |\bra{x} U_{y} \ket{0}|^{2}$ is obtained by noting that $\ket{\psi}=\frac{1}{\sqrt{2^{O(n^4)-n}}}\sum_{y} \ket{y} \otimes U_{y}\ket{0}$ (see Equation(\ref{eq9})), where $\ket{y}$ is a string of measurement results of non-output qubits sampling the random unitary $U_{y} \in \mathcal{U}_E $ which is  applied to the $n$-qubit input state $\ket{0}$ now teleported to the output position. 



In order to relate this to hardness, we first note that by construction our graph gadgets $G_E$ give rise to universal sets under post-selection $\mathcal{U}_E$ in $U(2^n)$ \footnote{To see this, note for large enough $k$ in $B^k$ we can generate any unitary in $U(4)$ under post-selection, because of universality of $\mathcal{U_B}$. In particular, we can generate to arbitrary accuracy the universal gate sets in \cite{univ1,univ2} for example, and SWAP's which are needed for universal quantum computation on $U(2^n)$.}.
This fact means that outputs $x$ are $\sharp$P-hard to approximate up to relative error 1/4 + O(1) in worst-case \cite{ising,bremner}. 
In the language of our MBQC gadgets, this translates to the fact that for some $U_y$ $\in$ $\mathcal{U}_E$ there exists outputs $x$ such that approximating $\langle x|U_y |0\rangle^2$ up to a relative error of 1/4 + O(1) is $\sharp$P-hard. 
This property is often referred to as worst-case $\sharp$P hardness \cite{juan,bremner} (or, for brevity, worst-case hardness), and is usually taken as a stepping stone for claiming average-case hardness conjectures of the likes of Conjecture 2 stated below. 
Hence, to obtain a working hardness proof (see Sections 1.2 and  2.3), we assume the 2 following complexity theoretic conjectures hold:
\begin{enumerate}
\item \emph{Conjecture 1}: The widely believed conjecture that the polynomial hierarchy (PH) does not collapse to its 3rd level. \cite{terhal}
\item \emph{Conjecture 2}: 	Approximating the output probabilities $\frac{1}{2^{O(n^4)-n}} |\bra{x} U_{y} \ket{0}|^{2}$  up to relative error  $\frac{1}{4}+O(1)$ for a constant fraction of unitaries $U_{y} \in \mathcal{U}_{E}$ is $\sharp$P-hard.
\end{enumerate}
\emph{Conjecture 2} seems plausible because one can relate the sampling problem  $D(x,y)$ to an IQP* sampling problem \cite{hoban}, and thus associate to it an appropriate Ising partition function \cite{ising,goldberg} . These Ising partition functions are known to be $\sharp$P-hard to approximate in worst case up to relative error $\frac{1}{4}+O(1)$  for circuits which are universal under post selection \cite{ising,goldberg,bremner,juan}. 
In this way, \emph{Conjecture 2} can be viewed as an average-case complexity conjecture on the approximation of Ising partition functions which is present in the usual hardness proofs \cite{gao,juan,juan2,bremner}. \\

We are now ready to precisely state  our second main result in the form of the following theorem:
\begin{theorem} 
\label{th2}
 Assuming conjectures 1 and 2 hold, a classical computer cannot sample from the distribution  $D(x,y)$ ( Equation (\ref{eq15})), formed from the concatenation of sampling partially invertible universal sets described above,  up to  $ l_{1}$-norm error $\frac{1}{22}$ in time $poly(n)$. \\
\end{theorem}

Our last analytical contribution concerns the universality of sets associated with random unitary ensembles generated by non-adaptive fixed  $XY$ angle measurements on cluster states. 
As seen in \cite{mantri,peterdamian,mezher} and for example in Figure \ref{fig3}, non adaptive fixed $XY$ angle measurements on cluster states suffice for generating random unitary ensembles $\{p_i,U_i \in \mathcal{U}\}$, with $\mathcal{U}$ universal in $U(2^n)$. 
Here we show that this universality is $generic$, meaning that almost any assignment of non-adaptive fixed $XY$ angle measurements on cluster states gives random unitary ensembles with support on universal gate sets $\mathcal{U} \in U(2^n)$, when $n=2^\gamma$, where $\gamma$ is a positive integer. 

Our starting point is the random unitary ensemble,
 \begin{equation}
\label{eq8}
 CGEN=\{\frac{1}{2^{n}},  CZ_{1,2}...CZ_{n-1,n}(HZ(\alpha_{1}+m_1\pi ) \otimes....\otimes HZ(\alpha_{n}+m_n\pi)) \},
\end{equation}
with $m_i \in \{0,1\}$. We show that this is an $(\eta <1,t)$-tensor product expander (TPE)  \cite{harrow2des,wintonviola,mezher,hastings}, meaning that (see Equation (\ref{eq3}) )
\begin{equation}
\label{eq gentpe}
||M_t[\mu_{CGEN}]-M_{t}[\mu_{H}]||_{\infty} \leq \eta < 1.
\end{equation}

$CGEN$ in Equation (\ref{eq8}) can be generated by an $n$-row, 2-column cluster state with $n$ output qubits-the last column is the (unmeasured output), and with $n$ $XY$ plane measurement angles $\alpha_i$, 
see Figure \ref{figgen}. 
We denote the set $\mathcal{U_{CGEN}}=\{CZ_{1,2}...CZ_{n-1,n}(HZ(\alpha_{1}+m_1\pi ) \otimes....\otimes HZ(\alpha_{n}+m_n\pi))\}$.
As seen in \cite{harrow2des,wintonviola}, showing that Equation (\ref{eq gentpe}) holds amounts to showing that the set $\mathcal{U_{CGEN}}$ is a  universal set in $U(2^{n})$ \cite{lloyd,barenco,brylinski}.  Our result about the universality of $\mathcal{U_{CGEN}}$ can be summarized in the following theorem.

\begin{theorem} 
\label{th3}
 $\mathcal{U_{CGEN}}$ is a universal set in $U(2^{n})$ for almost all choices of $\alpha_{1},...,\alpha_{n}$, for $n=2^{\gamma}$, where $\gamma$ is a positive integer.
\end{theorem}

Two immediate corrolaries follow from Theorem \ref{th3} and the results of \cite{wintonviola,harrow2des}.

\begin{corollary}
CGEN is an $(\eta <1,t)$-TPE for almost all choices of $\alpha_{1},...,\alpha_{n}$.
\end{corollary}

\begin{corollary}
$CGEN^{k}$ is an $\varepsilon$-approximate $t$-design for almost all choices of $\alpha_{1},...,\alpha_{n}$, and sufficiently large $k$.
\end{corollary}

$CGEN^{k}$ can be easily seen to generated by an $n$ row, $k+1$ column cluster state, with measurement angles $\alpha_i$, as illustrated in Figure \ref{fig7}.

A particularly interesting observation is the case when $\gamma=1$. The result of Theorem \ref{th3} in this case says that almost any 2-qubit cluster state gadgets $G_B$ generate random unitary ensembles $B$,  with universal sets $\mathcal{U_B} \subset U(4)$ \footnote{This is not surprising, since it was shown in \cite{lloyd,barenco} that almost any 2-qubit gate is universal for quantum computing.}. Where $\mathcal{U_B}$ can be  invertible, partially invertible, or non-invertible \footnote{ We mean by non-invertible that for all $U \in \mathcal{U_B}$ , $U^{\dagger} \notin \mathcal{U_B}$; We mean by invertible that for all $U \in \mathcal{U_B}$, $U^{\dagger} \in \mathcal{U_B}$ }. What remains in order to obtain \emph{efficient} $t$-designs is to show that the moment superoperator $M_t[\mu_B]$ of $B$ has a subdominant (second largest) eigenvalue $\lambda$,  and  $Conjecture$ $A$:  $|\lambda|$ does not scale badly (inefficiently) with $t$.  If $Conjecture$ $A$ is true, then we can apply the techniques  we used in  Theorem \ref{th1} to show that we can construct  $n$-qubit cluster state gadgets $LG_{block(B^k)}$ which sample from $t$-designs for efficient $L$ and $k$ from almost all 2-qubit cluster state gadgets $G_B$. Then, as a consequence of Theorem \ref{th2}, these $n$-qubit cluster state gadgets can be used in quantum speedup proposals. 
\begin{figure}[H]
\begin{center}
\graphicspath{}
\includegraphics[trim={10 4cm 100 0cm} , scale=0.3]{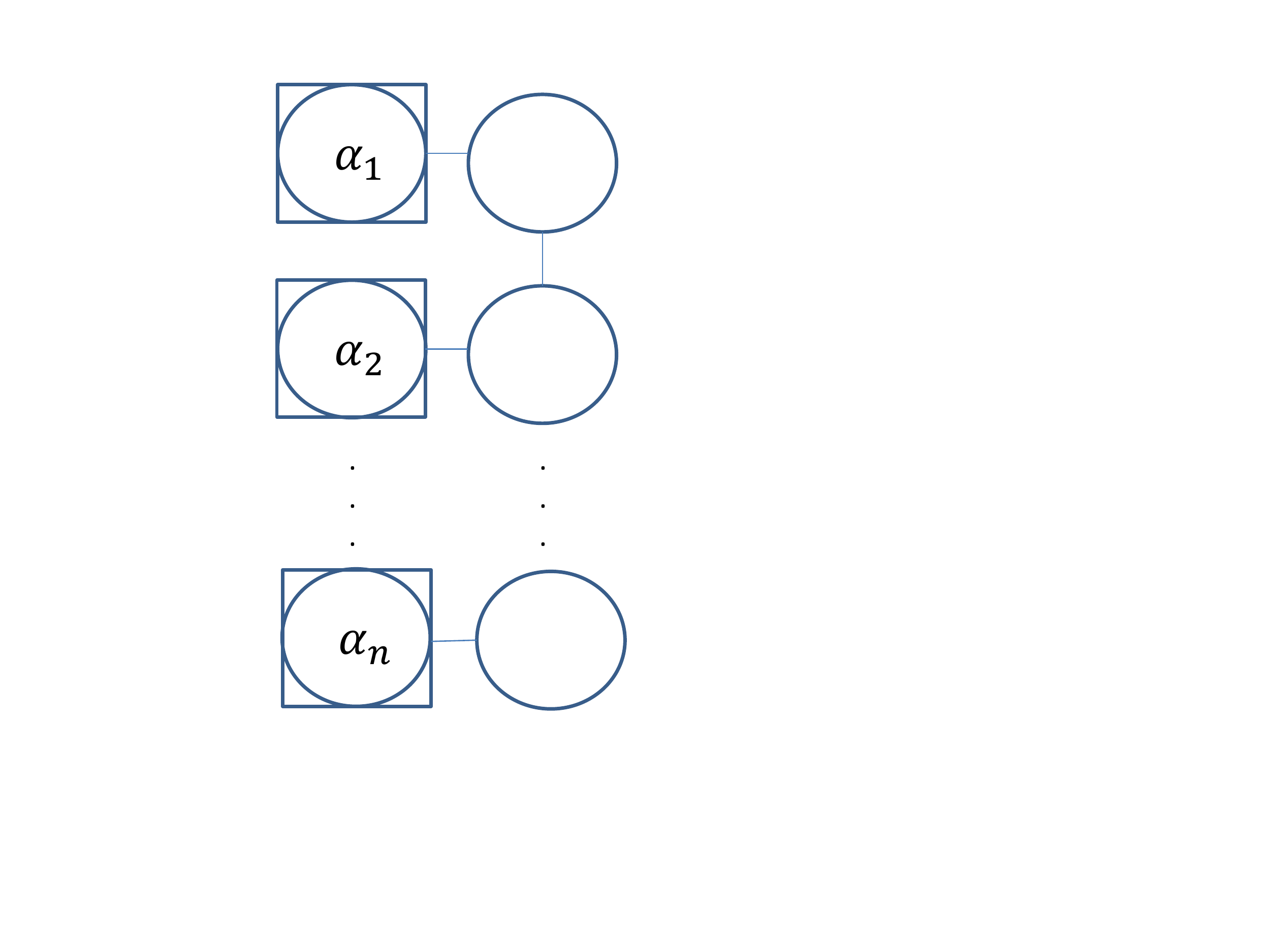}
\caption{ Cluster state gadget generating $CGEN$. Corollary 3 states that almost any choice of measurement angles $\alpha_{1},..\alpha_{n}$ give rise to a TPE.}
\label{figgen}
\end{center}
\end{figure}

\begin{figure}[H]
\begin{center}
\graphicspath{}
\includegraphics[trim={2 0cm 0 0cm} , scale=0.3]{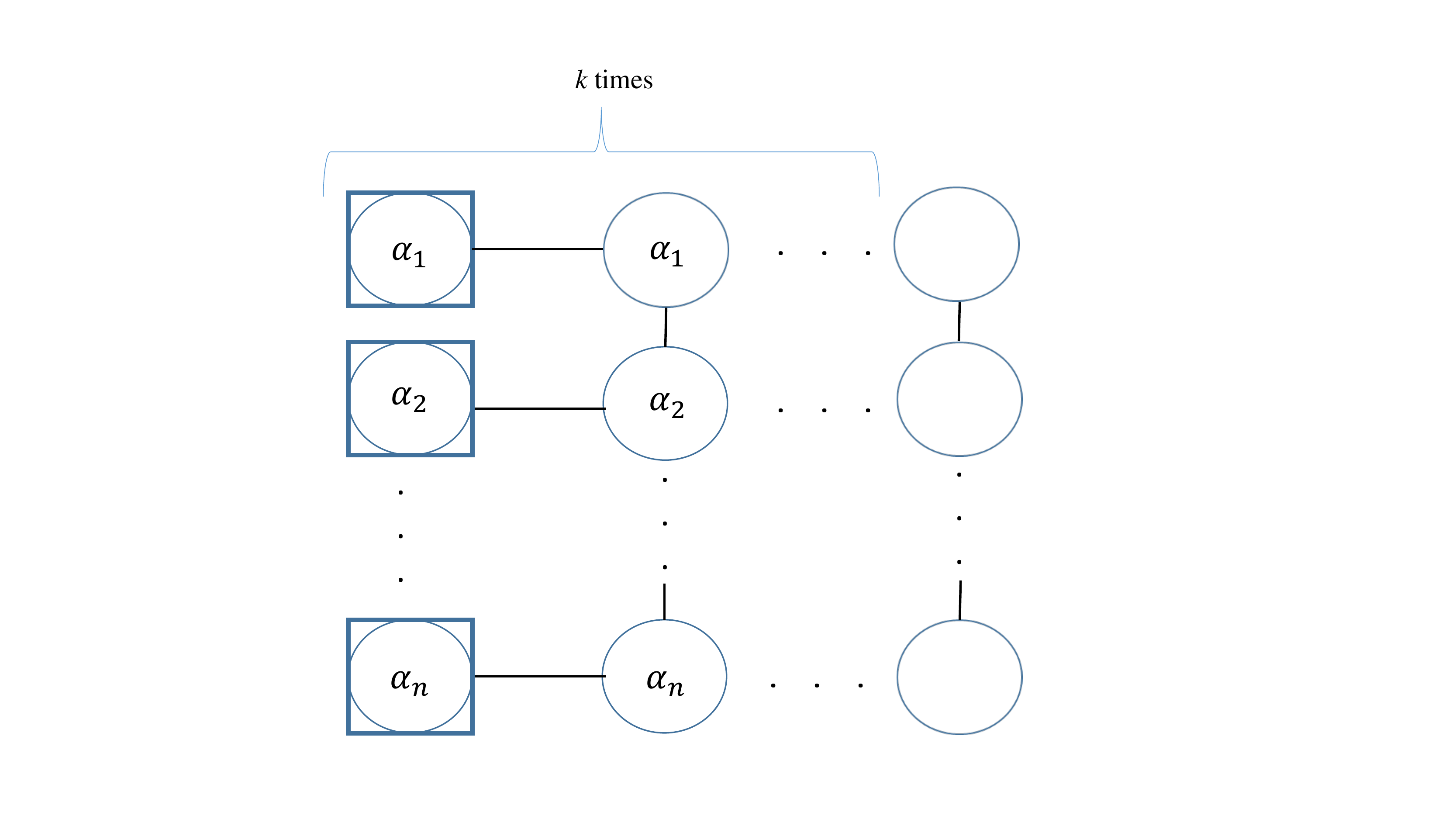}
\caption{ Graph gadget giving rise to $CGEN^k$. Corollary states that almost any choice of measurement angles $\alpha_{1},..\alpha_{n}$ give rise to a $t-design$. Numerics suggest this is also an efficient construction.}
\label{fig7}
\end{center}
\end{figure}

Concerning $Conjecture$ $A$, we performed numerical calculations on linear cluster states composed of 3 qubits, and on 2-row, 2-column cluster states like those of Figure \ref{fig1}. The random unitary ensembles of the 3  qubit linear cluster state have the form $\{\frac{1}{4},HZ^mZ(\alpha)HZ^{m^{'}}Z(\alpha)\in U(2) \}$ where $m,m^{'} \in \{0,1\}$. These random ensembles are generated by measuring two of the qubits of the linear cluster state at an angle $\alpha$ in the XY plane. The random unitary ensembles corresponding to the 2-row, 2- column cluster states have the form of Equation (\ref{eq10}), and are generated by XY plane measurements performed as in Figure \ref{fig1}. The numerics are based on calculating the subdominant eigenvalue $|\lambda|$ of the moment superoperator (see Definition \ref{defmomentsuperop}) corresponding to each of the above random unitary ensembles, for various values of $t$, and for various choices of the XY plane measurement angles. For the 3 qubit linear cluster states the values of $t$ tested were $t=2,3,4,5$, and for the 2-row, 2-column cluster states we tested for $t=2,3$. Beyond these values the numerical investigation becomes unfeasible as our numerical algorithms scale exponentially with $n$ and $t$.  \footnote{ Note that in the $t=1$ case, we obtained exact 1-designs ($|\lambda|=0$) for both linear cluster states and 2-row, 2-column cluster states. This is in line with numerical calculations performed in \cite{peterdamian}.}. For all the choices of fixed angle, non-adaptive $XY$ measurements  tested, we found that the subdominant eigenvalue $|\lambda|$ was independent of $t$ for $t=2,3$ for both the 3 qubit linear cluster states and the 2-column, 2-row cluster states, which is in line with calculations in \cite{CHM+13}. On the other hand, for the 3 qubit linear cluster states some angles tested showed a $|\lambda|$ $independent$ of $t$ for $t=2,3,4,5$, which is in line with the result of \cite{bourgaingambourd}, other angles showed that $|\lambda|$ changes values from $t=3$ to $t=4$, but remains the same for $t=4$ and $t=5$. These numerical calculations seem to confirm $Conjecture$ $A$. (see Appendix for further discussion)
As a final remark, note that in our numerics we assume $\eta \sim |\lambda|$ (see Definition \ref{deftpe}) for moment superoperators of random ensembles defined on universal sets $\mathcal{U}$. We mean by this that the rate at which $t$-designess is attained is determined asymptotically by $|\lambda|$. This is indeed true, and common practice, when this moment superoperator is Hermitian and, more importantly, diagonalizable \cite{wintonviola}. This corresponds to the case when $\mathcal{U}$ contains unitaries and their inverses. However, this approximation also works for general moment superoperators $M_t[\mu]$, namely because the set of diagonalizable square N by N matrices is \emph{dense} in the set of N by N square matrices \cite{linear}. This means that any non-diagonalizable  $M_t[\mu]$ is arbitrarily close in norm to a diagonal matrix, and in particular their eigenvalues are arbitrarily close.

\section{Proof of Theorems}
\subsection{Proof of Theorem \ref{th1}}
We begin by proving the following Lemma regarding the emsemble $B$ which samples from the partially invertible set $\mathcal{U_B}$ (see Equation (\ref{eqpartinv})).
\begin{lemma}
\label{lem1}
 B is an $(\eta,t)$-TPE with $\eta=1+(C-1)a <1$ where $0<C<1$, and $a=\frac{|\mathcal{U_M}|}{|\mathcal{U_B}|}$.\\
\end{lemma}
\begin{proof}
$$M_t[\mu_B]=\sum_{i=\{1,...,|\mathcal{U_B}|\}}\frac{1}{|\mathcal{U_B}|}U_i^{\otimes t,t}=aM_t[\mu_M]+(1-a)M_t[\mu_{B/M}]$$ \\ where $$M_t[\mu_M]=\sum_{i=\{1,...,|\mathcal{U_M}|\}}\frac{1}{|\mathcal{U_M}|}U_i^{\otimes t,t}, U_i \in \mathcal{U_M},$$   and $$M_t[\mu_{B/M}]=\sum_{i=\{1,...,|\mathcal{U_{B/M}}|\}}\frac{1}{|\mathcal{U_{B/M}}|}U_i^{\otimes t,t}, U_i \in \mathcal{U_{B/M}}.$$
Since, by our definition of a partially invertible universal set, $\mathcal{U_{B/M}}$ is universal in $U(4)$, meaning by Proposition \ref{prp2}, that \cite{harrow2des}
\begin{equation}
\label{eq1th1}
||M_t[\mu_{B/M}]-M_t[\tilde \mu_H]||_\infty \leq 1 
\end{equation}
where $\tilde \mu_H$ is the Haar measure on $U(4)$ (as opposed to $\mu _H$ in Equation (\ref{eq1}) which refers to the Haar measure over $U(2^n)$), and $M_t[\tilde \mu_H]=\int_{U(4)} U^{\otimes t,t} \tilde \mu_H(dU)$. Furthermore, $M_t[\mu_M]$ is the moment superoperator associated to a random ensemble sampling uniformly from a universal set in $U(4)$ having unitaries with algebraic entries \footnote{In \cite{bourgaingambourd}, one  requires sampling from SU(4). Fortunately, the moment super operator of a set sampled from U(4) can always be thought of as a sampling from SU(4). This can be seen by noting that for all $U$ $\in$ U(4) we have $det(U)$ $\neq$ 0, hence $U^{\otimes t,t }$=$|det(U)|^{\dfrac{t}{2}}$.$U^{ ' \otimes t,t}$=$U^{ ' \otimes t,t}$,  where $U^{'}$ $\in$ SU(4).}, and which contains inverses, $M=\{\frac{1}{|\mathcal{U_M}|},U_{i} \in \mathcal{U_M}\}$. Then, from the result of \cite{bourgaingambourd}, there is a constant $0 <C <1$ independent of $t$ such that the following relation holds 
\begin{equation}
\label{eq5th1}
||M_t[\mu_{M}]-M_t[\tilde \mu_H]||_\infty \leq C. \\
\end{equation}
Now, 
$$||M_t[\mu_B]-M_t[\tilde \mu_H]||_\infty = ||aM_t[\mu_{M}]-aM_t[\tilde \mu_H]+(1-a)M_t[\mu_{B/M}]-(1-a)M_t[\tilde \mu_H]||_\infty ,$$ 
thus
\begin{equation}
\label{eq6th1}
||M_t[\mu_B]-M_t[\tilde \mu_H]||_\infty \leq a||M_t[\mu_{M}]-M_t[\tilde \mu_H]||_\infty +(1-a) ||M_t[\mu_{B/M}]-M_t[\tilde \mu_H]||_\infty=\eta.
\end{equation}
Replacing Equations (\ref{eq5th1}) and (\ref{eq1th1}) in Equation (\ref{eq6th1}) allows to obtain the desired value of $\eta$. \\
\end{proof}

Using Proposition \ref{prop1} and Lemma \ref{lem1} we have the direct corollary concerning the $k$-fold concatenation of $B$, denoted by $B^k$ (see Equation (\ref{eqbk})).
\begin{corollary}
\label{cortdes}
$B^k$ is a $\varepsilon$-approximate $t$-design in $U(4)$ for $k \geq \frac{1}{log_{2}(\frac{1}{1+(C-1)a})}(8t+log_{2}(\frac{1}{\varepsilon})).$
\end{corollary}
The next step is to consider the random unitary ensemble $block(B^{k})$ (Equation (\ref{eqblockbk})) whose associated moment superoperator is $M_t[\mu_{block(B^k)}]$. We will prove the following Lemma.
\begin{lemma}
\label{lem2}
$M_t[\mu_{block(B^k)}]=P^{'}_{even}P^{'}_{odd}$
\end{lemma}
where
$$P^{'}_{even}=P^{'}_{2,3}.P^{'}_{4,5}...,$$  $$P^{'}_{odd}=P^{'}_{1,2}.P^{'}_{3,4}...,$$  and $$P^{'}_{i,i+1}= \frac{1}{|\mathcal{U}_{\mathcal{B}^k}|}\sum_{j=\{1,...,|\mathcal{U}_{\mathcal{B}^k}|\}}(1_{2 \times 2}^{\otimes i-1} \otimes U^{j}_{i,i+1} \otimes 1_{2 \times 2}^{\otimes n-i-1})^{\otimes t,t},$$  where $U^{j}_{i,i+1} \in \mathcal{U}_{\mathcal{B}^k}$.
\begin{proof}
$$block(B^k)=\{\frac{1}{|\mathcal{U}_{\mathcal{B}^k}|^{n-1}} , (1_{2 \times 2} \otimes U^{j_1}_{2,3} \otimes U^{j_2}_{4,5}\otimes ...\otimes U^{j_{\frac{n}{2}-1}}_{n-2,n-1} \otimes 1_{2 \times 2}). (U^{j_\frac{n}{2}}_{1,2}\otimes U^{j_{\frac{n}{2}+1}}_{3,4}\otimes ...\otimes U^{j_{n-1}}_{n-1,n})\},$$ where $U^{j}_{i,i+1} \in \mathcal{U}_{\mathcal{B}^k}.$ \\ \\
$M_t[\mu_{block(B^k)}]=\\ \sum_{j_1,j_2,..,j_{n-1}=1,...|\mathcal{U}_{\mathcal{B}^k}|}\frac{1}{|\mathcal{U}_{\mathcal{B}^k}|^{n-1}}\left ( (1_{2 \times 2} \otimes U^{j_1}_{2,3} \otimes U^{j_2}_{4,5}\otimes ...\otimes U^{j_{\frac{n}{2}-1}}_{n-2,n-1} \otimes 1_{2 \times 2}). (U^{j_\frac{n}{2}}_{1,2}\otimes U^{j_{\frac{n}{2}+1}}_{3,4}\otimes ...\otimes U^{j_{n-1}}_{n-1,n}) \right )^{\otimes t,t}.$\\ \\
$M_t[\mu_{block(B^k)}]$ can be rewritten as :\\ \\
$M_t[\mu_{block(B^k)}]=\left [\left (\frac{1}{|\mathcal{U}_{\mathcal{B}^k}|}\sum_{j_1=1,...|\mathcal{U}_{\mathcal{B}^k}|}(1_{2 \times 2} \otimes  U^{j_1}_{2,3} \otimes 1_{2 \times 2}^{\otimes n-3})^{\otimes t,t} \right ) \times \left( \frac{1}{|\mathcal{U}_{\mathcal{B}^k}|}\sum_{j_2=1,...|\mathcal{U}_{\mathcal{B}^k}|}( 1_{2 \times 2}^{\otimes 3} \otimes U^{j_2}_{4,5} \otimes 1_{2 \times 2}^{\otimes n-5})^{\otimes t,t}\right ) ...\right ] \times \\
\left [ \frac{1}{|\mathcal{U}_{\mathcal{B}^k}|}\sum_{j_{\frac{n}{2}}=1,...|\mathcal{U}_{\mathcal{B}^k}|}(  U^{j_\frac{n}{2}}_{1,2}\otimes 1_{2 \times 2}^{\otimes n-2})^{\otimes t,t} ...\right ]=[P^{'}_{2,3}.P^{'}_{4,5}...].[P^{'}_{1,2}.P^{'}_{3,4}...]=P^{'}_{even}P^{'}_{odd}.$
\end{proof}
Next, we would like to bound $||P^{'}_{even}P^{'}_{odd}- P^{H}_{even}.P^{H}_{odd}||_\infty$, where $$P^{H}_{even}=P^{H}_{2,3}.P^{H}_{4,5}...,$$  $$P^{H}_{odd}=P^{H}_{1,2}.P^{H}_{3,4}...,$$  and $$P^{H}_{i,i+1}=\int_{U(4)} (1_{2 \times 2}^{\otimes i-1} \otimes  U  \otimes 1_{2 \times 2}^{\otimes n-i-1})^{\otimes t,t} \tilde \mu_H(dU) .$$ We start by bounding each $P^{'}_{i,i+1}$ individually. Recall the 2 following well known and easily provable facts. \emph{Fact 1}: for complex $N$ by $N$ matrices $A$ we have 
\begin{equation}
\label{eqfact1}
\frac{1}{\sqrt{N}}||A||_\infty \leq ||A||_2 \leq \sqrt{N} ||A||_\infty.
\end{equation}
\emph{Fact 2} : For Complex matrices $A$ and $B$, 
\begin{equation}
\label{eqfact2}
|| A \otimes B||_2 =||A||_2.||B||_2. 
\end{equation}
\\ Now, $$||P^{'}_{i,i+1}-P^{H}_{i,i+1}||_\infty \leq 2^{nt}.||P^{'}_{i,i+1}-P^{H}_{i,i+1}||_2 \leq 2^{nt}  || \frac{1}{|\mathcal{U}_{\mathcal{B}^k}|}\sum_{j=\{1,...,|\mathcal{U}_{\mathcal{B}^k}|\}}(U^{j}_{i,i+1})^{\otimes t,t}-\int_{U(4)}U^{\otimes t,t} \tilde \mu_H(dU) ||_2.$$ The rightmost term is obtained using \emph{Fact 2} (Equation (\ref{eqfact2})) and noting that $|| 1 ||_{2}=1$. Using \emph{Fact 1} (Equation (\ref{eqfact1})) again, we get:  $$||P^{'}_{i,i+1}-P^{H}_{i,i+1}||_\infty \leq 2^{nt+t}|| \frac{1}{|\mathcal{U}_{\mathcal{B}^k}|}\sum_{j=\{1,...,|\mathcal{U}_{\mathcal{B}^k}|\}}(U^{j}_{i,i+1})^{\otimes t,t}-\int_{U(4)}U^{\otimes t,t} \tilde \mu_H(dU) ||_\infty.$$ Note that, $$M_t[\mu_B^k]=\frac{1}{|\mathcal{U}_{\mathcal{B}^k}|}\sum_{j=\{1,...,|\mathcal{U}_{\mathcal{B}^k}|\}}(U^{j}_{i,i+1})^{\otimes t,t},$$ and $$ M_t[ \tilde \mu_H]=\int_{U(4)}U^{\otimes t,t} \tilde \mu_H(dU) .$$ \\ Now, because $B^{k}$ is a $\varepsilon$-approximate $t$-design on $U(4)$ (see Corollary \ref{cortdes}), we have from \cite{brandao} that: $$|| M_t[\mu_B^k]-M_t[ \tilde \mu_H]||_\infty \leq 2^{t+1} \varepsilon.$$ Substituting this inequality in $||P^{'}_{i,i+1}-P^{H}_{i,i+1}||_\infty$ gives,
\begin{equation}
\label{eqboundonp}
||P^{'}_{i,i+1}-P^{H}_{i,i+1}||_\infty \leq 2^{nt+2t +1} \varepsilon.
\end{equation}
Choosing $\varepsilon=\frac{\varepsilon_1}{2^{nt+2t+1}}$ we get that,
\begin{equation}
\label{newboundonp}
||P^{'}_{i,i+1}-P^{H}_{i,i+1}||_\infty \leq \varepsilon_1
\end{equation}

when
\begin{equation}
\label{newboundk}
 k \geq \frac{1}{log_{2}(\frac{1}{1+(C-1)a})}(10t+nt+1+log_{2}(\frac{1}{\varepsilon_1})).
\end{equation}
Equation (\ref{newboundk}) is found by plugging the value of  $\varepsilon$ in Corollary \ref{cortdes}.
Now we are ready to bound $$||P^{'}_{even}P^{'}_{odd}- P^{H}_{even}.P^{H}_{odd}||_\infty.$$ We claim
\begin{lemma}
$||P^{'}_{even}P^{'}_{odd}- P^{H}_{even}.P^{H}_{odd}||_\infty \leq2^{n^2t-2nt+n-1}\varepsilon_1$.
\end{lemma}
\begin{proof} 
From Equation (\ref{newboundonp}), we can write for all $i$ ,  $P^{'}_{i,i+1}=P^{H}_{i,i+1} + \gamma_{i}$. where $||\gamma_{i}||_\infty \leq \varepsilon_1$.
$$||P^{'}_{even}P^{'}_{odd}- P^{H}_{even}.P^{H}_{odd}||_\infty =||(P^{H}_{1,2}+\gamma_{1})(P^{H}_{3,4}+\gamma_{3})...-P^{H}_{1,2}P^{H}_{3,4}...||_\infty.$$
Thus,
$$||(P^{H}_{1,2}+\gamma_{1})(P^{H}_{3,4}+\gamma_{3})...-P^{H}_{1,2}P^{H}_{3,4}...||_\infty =  ||P^{H}_{1,2}P^{H}_{3,4}..+P^{H}_{1,2}\gamma_{3}..+\gamma_{1}P^{H}_{3,4}..+\gamma_{1}\gamma_{3}...-P^{H}_{1,2}P^{H}_{3,4}...||_\infty.$$\\ 
Thus, $$||(P^{H}_{1,2}+\gamma_{1})(P^{H}_{3,4}+\gamma_{3})...-P^{H}_{1,2}P^{H}_{3,4}...||_\infty\leq ||P^{H}_{1,2}\gamma_{3}..||_\infty +||\gamma_{1}P^{H}_{3,4}..||_\infty+||\gamma_{1}\gamma_{3}..||_\infty+...$$ \\
$||P^{H}_{1,2}\gamma_{3}..||_\infty +||\gamma_{1}P^{H}_{3,4}..||_\infty+||\gamma_{1}\gamma_{3}..||_\infty+...$ is a sum of $2^{n-1}-1$ terms, each containing at most a product of $n-2$ $P^{H}_{i,i+1}$'s. Noting that, $|| P^{H}_{i,i+1}||_\infty \leq 2^{nt} ||P^{H}_{i,i+1}||_2$  using \emph{Fact 1} (Equation (\ref{eqfact1})), and - using \emph{Fact 2} (Equation (\ref{eqfact2}))- that $||P^H_{i,i+1}||_2=||M_t[\tilde \mu_H]||_2=1$, then every term of the sum is individually less than $(2^{nt})^{n-2}\varepsilon_1 $ \footnote{Noting that $\varepsilon_1 < 1$, so $\varepsilon^{m}_1 <\varepsilon_1$ for all $m>1$}, which means the whole sum (i.e $||P^{'}_{even}P^{'}_{odd}- P^{H}_{even}.P^{H}_{odd}||_\infty$)  is less than $(2^{n-1}-1)(2^{nt})^{n-2}\varepsilon_1$, or equivalently less than $2^{n^2t-2nt+n-1}\varepsilon_1$.
\end{proof}

Again, choosing $\varepsilon_{1}=\frac{\varepsilon^{'}}{2^{n^2t-2nt+n-1}}$, we get 
\begin{equation}
\label{eqess}
||P^{'}_{even}P^{'}_{odd}- P^{H}_{even}.P^{H}_{odd}||_\infty \leq \varepsilon^{'}
\end{equation}

when 
\begin{equation}
\label{finalk}
k \geq \frac{1}{log_{2}(\frac{1}{1+(C-1)a})}(10t+n^2t-nt+n+log_{2}(\frac{1}{\varepsilon^{'}})).
\end{equation}
Finally, we prove the following Lemma.
\begin{lemma}
\label{tpe}
For $n \geq \lfloor{2.5log_2(4t)} \rfloor$, $block(B^k)$ is an $(\eta,t)$-TPE on $U(2^n)$ with $\eta=P(t)+\varepsilon^{'}$. Where $P(t)$ is a polynomial in $t$ given by Equation (\ref{eqpt})
\end{lemma}
\begin{proof}
We need to bound $||M_t[\mu_{block(B^k)}]-M_t[\mu_H]||_\infty$, where $M_t[\mu_H]=\int_{U(2^n)} U^{\otimes t,t} \mu_H[dU]$, and $\mu_H$ is the Haar measure on $U(2^n)$. from Lemma \ref{lem2},  $$||M_t[\mu_{block(B^k)}]-M_t[\mu_H]||_\infty=||P^{'}_{even}P^{'}_{odd}-M_t[\mu_H]||_\infty.$$ By a triangle inequality, $$||M_t[\mu_{block(B^k)}]-M_t[\mu_H]||_\infty \leq ||P^{H}_{even}P^{H}_{odd}-M_t[\mu_H]||_\infty+||P^{'}_{even}P^{'}_{odd}- P^{H}_{even}.P^{H}_{odd}||_\infty.$$  Plugging in Equation (\ref{eqess}) we get : $$||M_t[\mu_{block(B^k)}-M_t[\mu_H]||_\infty \leq  ||P^{H}_{even}P^{H}_{odd}-M_t[\mu_H]||_\infty+\varepsilon^{'}.$$ Finally, from the Detectibility Lemma \cite{aharonov} and the result of \cite{brandao} we get that when $$n \geq \lfloor{2.5log_2(4t)} \rfloor,$$  $$||P^{H}_{even}P^{H}_{odd}-M_t[\mu_H]||_\infty \leq (1+\frac{(425\lfloor{log_2(4t)} \rfloor^2 t^5 t^{3.1/log(2)})^{-1}}{2})^{-1/3}:=P(t),$$ and hence, $$||M_t[\mu_{block(B^k)}]-M_t[\mu_H]||_\infty \leq P(t)+\varepsilon^{'}.$$
\end{proof}
Using Lemma \ref{tpe} and Proposition \ref{prop1} one obtains directly the value of $L$ in Theorem \ref{th1} with $k$ given by Equation (\ref{finalk}), and $n \geq \lfloor{2.5log_2(4t)} \rfloor$. This Completes our proof of Theorem \ref{th1}.

\subsection{Proof of Theorem \ref{th2}}
We will follow the standard technique of applying Stockmeyer’s theorem \cite{stockmeyer} along with some average-case hardness conjecture \cite{gao,juan,juan2,bremner} to prove hardness of approximate classical sampling up to a constant $ l_{1}$-norm error . These techniques are the same as those used in \cite{juan,juan2}. 
In our proof we will rely only on the 2 conjectures mentioned in Section 3. \\ 

 Let $D(x,y)$ be the distibution given by probabilities $p(x,y)=|\langle{x,y} | \psi \rangle|^{2}$ as defined in Equation (\ref{eq15}). Suppose there exists a classical $poly(n)=poly(O(n^4))$- time  algorithm C which can sample from a probability distribution that approximates $D(x,y)$ up to an additive error $\mu$ in $ l_{1}$-norm. In other words (following Equation (\ref{eq class sim}) ) :
\begin{equation}
\label{teq6}
\sum_{x,y}|p(x,y)-p_{C}(x,y)| \leq \mu .
\end{equation}
$p_{C}(x,y)$ is the output probability of the classical algorithm $C$. 
Then by Stockmeyer’s theorem \cite{stockmeyer} there exists an $ FBPP^{NP}$    algorithm that computes an estimate $p_{C}^{\sim}(x,y)$ of $p(x,y)$  such that:
\begin{equation}
\label{teq7}
|p_{C}^{\sim}(x,y)-p(x,y)| \leq\frac{p(x,y)}{poly(O(n^4))} \\+|p_{C}(x,y)-p(x,y)| (1+\frac{1}{poly(O(n^4))}) .
\end{equation}

Using Markov’s inequality:
\begin{equation}
\label{teq8}
Pr_{x,y}(|p_{C}(x,y)-p(x,y)| \geq \\  \frac{E(|p_{C}(x,y)-p(x,y)|)}{\delta})   \leq \delta , 
\end{equation}

where $0<\delta\leq 1$ and $\ket{x,y}$ picked uniformly at random. Noting that Equation (\ref{teq6})  implies
$E(|p_{C}(x,y)-p(x,y)|) \leq \frac{\mu}{2^{O(n^4)}}$

 we get:
\begin{equation}
\label{teq9}
Pr_{x,y}(|p_{C}(x,y)-p(x,y)| \geq \frac{\mu}{\delta2^{O(n^4)}}) \leq \delta .
\end{equation}

Equation (\ref{teq9}) means that the following relation holds with probability $1-\delta$:
\begin{equation}
\label{teq10}
|p_{C}^{\sim}(x,y)-p(x,y)| \leq\frac{p(x,y)}{poly(O(n^4))} \\+\frac{\mu}{\delta2^{O(n^4)}} .
\end{equation}

We now use the following anti-concentration property for 2-designs (see Equation (\ref{eq anticonc})):
\begin{equation}
\label{teq11}
Pr_{U_{y} \sim \mu} (|\bra x U_{y} \ket 0|^{2} > \frac{\alpha(1-\varepsilon_d)}{2^{n}} ) \geq \frac{(1-\alpha)^{2} (1-\varepsilon_d)}{2(1+\varepsilon_d)} ,
\end{equation}

where $ 0<\alpha \leq 1$ . Note that measurement of the $O(n^4)-n$  non-output qubits simply induces a uniform $\frac{1}{2^{O(n^4)-n} }$ distribution, and one can recast Equation (\ref{teq11}) to reflect anti-concentration on the entire $O(n^4)$ measured qubits:
\begin{equation}
\label{teq12}
Pr_{U_{y} \sim \mu} (p(x,y) > \frac{\alpha(1-\varepsilon_d)}{2^{O(n^4)}} ) \geq \frac{(1-\alpha)^{2} (1-\varepsilon_d)}{2(1+\varepsilon_d)} .
\end{equation}

Equation (\ref{teq12}) implies:
\begin{equation}
\label{teq13}
\frac{\mu}{\delta 2^{O(n^4)}} \leq \frac{\mu}{\delta \alpha (1-\varepsilon_d)}p(x,y) .
\end{equation}
                                                                                     
Equation (\ref{teq13}) holds with probability $\frac{(1-\alpha)^{2} (1-\varepsilon_d)}{2(1+\varepsilon_d)}$  .
Plugging Equation (\ref{teq13}) into Equation (\ref{teq10}) we obtain:
\begin{equation}
\label{teq14}
|p_{C}^{\sim}(x,y)-p(x,y)| \leq  (O(1)+\frac{\mu}{\delta \alpha (1-\varepsilon_d)})p(x,y) .
\end{equation}
Equation (\ref{teq14}) is an approximation of  $p(x,y)$  by $p_{C}^{\sim}(x,y)$  with  relative error $O(1)+\frac{\mu}{\delta \alpha (1-\varepsilon_d)}$.We claim, by a similar reasoning as can be found in \cite{bremner, juan}, that Equation (\ref{teq14}) holds with probability $(1-\delta)\frac{(1-\alpha)^{2} (1-\varepsilon_d)}{2(1+\varepsilon_d)}$, or in other words Equation (\ref{teq14}) is true for a $(1-\delta)\frac{(1-\alpha)^{2} (1-\varepsilon_d)}{2(1+\varepsilon_d)}$ fraction of unitaries $U_{y} \in \mathcal{U}_E$. 
Choosing $\mu=\frac{1}{22}$,  $\delta=\alpha= \varepsilon_d \sim$0.1132, we get that $p_{C}^{\sim}(x,y)$  approximates $p(x,y)$ to a relative error of  $\frac{1}{4}+O(1)$ for an $\sim$ 0.28 fraction of unitaries $U_{y} \in \mathcal{U}_{E}$ .
Assuming\emph{ Conjecture 2} to be true, we now have an $FBPP^{NP}$ algorithm which solves a $\sharp$ P-hard problem. But, this would imply by Toda’s theorem \cite{toda} that the PH collapses to its 3rd level. Because we conjecture (\emph{Conjecture 1}) the PH collapse to be impossible, we thus obtain a contradiction. As a conclusion, $ D(x,y)$ cannot be sampled from up to a constant $l_{1}$-norm error by a classical polynomial time algorithm . This concludes our proof of Theorem \ref{th2}.

\subsection{Proof of Theorem \ref{th3}}
We start with $\gamma=1$,  then  $$\{p_{i},U_{i}\}=\{\frac{1}{4},  CZ(HZ(\alpha_{1}+m_1\pi) \otimes HZ(\alpha_{2}+m_2\pi))\},$$ where $m_1, m_2 \in \{0,1\}$, and $$\mathcal{U_{CGEN}}=\{CZ(HZ(\alpha_{1}+m_1\pi) \otimes HZ(\alpha_{2}+m_2\pi))\}.$$    We  suppose $\alpha_1 \in [0,2\pi]$ and $\alpha_2 \in [0,2\pi]$ are fixed angles irrationally related to $\pi$. Note that $almost$ $any$ angle is irrationally related to $\pi$, meaning that the set of angles rationally related to $\pi$ in the interval $[0,2\pi]$ have Lebesgue measure zero \cite{lebesgue} \footnote{Note also that the Lebesgue measure of the  set of all points of the form $\{\alpha_1,...,\alpha_n\}$, where each of the $\alpha_i$'s are rationally related to $\pi$ is also zero. That's because the Lebesgue measure of a  cartesian product of sets is equal to the product of Lebesgue measures of individual sets, and each of the individual sets (i.e a set of angles which is rationally related to $\pi$) has Lebesgue measure zero \cite{lebesgue}.}. Denote by $Lie(U(4))$ the Lie algebra of $U(4)$  and $Lie(U(2) \otimes U(2))$ that of $U(2) \otimes U(2)$ \cite{lie} \footnote{We mean by this that $Lie(U(2) \otimes U(2))$ is the Lie algebra of unitary matrices $S \otimes T$. Where $S,T \in U(2)$}. We want to prove, following \cite{lloyd,barenco}, that one can find at least two unitaries $U_1$ and $U_2$  in the random ensemble that have eigenvalues having arguments irrationally related to $\pi$ and whose Lie algebra is a generic element of $Lie(U(4))$, and not any subalgebra. In that way we can construct any element of $U(4)$ from products of $U_1$ and $U_2$ \cite{lloyd,barenco}. For our purposes, consider $$U_1=CZ_{1,2}(HZ(\alpha_{1}) \otimes HZ(\alpha_{2})),$$ and $$U_2=CZ_{1,2}(HZ(\alpha_{1}+\pi) \otimes HZ(\alpha_{2}+\pi)).$$ The requirement of eigenvalues having arguments irrationally related to $\pi$ is fulfilled by our choice of angles. We still need to prove we can find unitaries whose Lie algebras are in $Lie(U(4))$ and not any subalgebra.We begin by proving the following lemma.
\begin{lemma}
\label{l1th3}
$\frac{log(HZ(\alpha_{1}) \otimes HZ(\alpha_{2}))}{i}  $ and $\frac{log(HZ(\alpha_{1}+\pi) \otimes HZ(\alpha_{2}+\pi))}{i}$ are generic elements of $Lie(U(2) \otimes U(2))$ for $\alpha_1$, $\alpha_2$ irrationally related to $\pi$.
\end{lemma}
\begin{proof}
It suffices to prove that $HZ(\alpha)$ (or equivalently $HZ(\alpha)HZ(\alpha)$ )is a generic element of $U(2)$ (and not any subgroup), for $\alpha$ generically chosen.
Direct calculation gives $$HZ(\alpha)HZ(\alpha)=e^{i\alpha}\begin{bmatrix} \frac{1+ e^{-i\alpha}}{2} \frac{1- e^{i \alpha}}{2} \\   \frac{e^{-i\alpha}-1}{2}  \frac{1+ e^{i\alpha}}{2} \end{bmatrix}$$ where  $$R=\begin{bmatrix} \frac{1+ e^{-i\alpha}}{2} \frac{1- e^{i \alpha}}{2} \\   \frac{e^{-i\alpha}-1}{2}  \frac{1+ e^{i\alpha}}{2} \end{bmatrix} \in SU(2) .$$  A well known fact about $SU(2)$ is that a generic element can be represented as $e^{i\delta \vec{n}\vec{\sigma}}$ \cite{nielsenchuang}. Where $\vec{n}=a\vec{x}+b\vec{y}+c\vec{z}$. $a$, $b$ and $c$ are real numbers such that $|a|^2+|b|^2+|c|^2=1$.\\  $\sigma=X\vec{x}+Y\vec{y}+Z\vec{z}$, $X$, $Y$ and $Z$ are the Pauli matrices. Again, a direct calculation for $R$ gives $\delta=cos^2(\frac{\alpha}{2})$, $a=c=-\frac{sin\alpha}{2\sqrt{1-cos^4(\frac{\alpha}{2})}}$, and $b=-\frac{1-cos\alpha}{2\sqrt{1-cos^4(\frac{\alpha}{2})}}$. None of $\delta$, $a$, $b$ or $c$ are zero for generically chosen $\alpha$, this means that $R$ is a generic element of $SU(2)$. Since $$det(HZ(\alpha)HZ(\alpha))=det(e^{i\alpha}\begin{bmatrix} \frac{1+ e^{-i\alpha}}{2} \frac{1- e^{i \alpha}}{2} \\   \frac{e^{-i\alpha}-1}{2}  \frac{1+ e^{i\alpha}}{2} \end{bmatrix})=e^{2i\alpha} \neq 1$$ for generically chosen $\alpha$, it means $HZ(\alpha)HZ(\alpha)$ (and hence $HZ(\alpha)$) is a generic element of $U(2)$ for generic $\alpha$.
\end{proof}
Now, since $CZ \notin Lie(U(2) \otimes U(2))$ because $CZ$ is not decomposable into a product of 1-qubit gates. Thus, $\frac{log(U_1)}{i}$ and $\frac{log(U_2)}{i} \in f$, where $Lie( U(2) \otimes U(2)) \subset f$. By  lemma 6.1 in \cite{brylinski} we have that there is no intermediate Lie algebra between $Lie(U(d) \otimes U(d))$ and $Lie(U(d^2)$, hence $f=Lie(U(4))$, and thus $\frac{log(U_1)}{i}$ and $\frac{log(U_2)}{i}$ are generic elements of $Lie(U(4))$ . This concludes the proof of the $\gamma=1$ case \footnote{ A similar proof of this is found in \cite{harrownote}, while noting that Lemma 5 along with results of \cite{lloyd,barenco} implies $< HZ(\alpha_{1}) \otimes HZ(\alpha_{2})),HZ(\alpha_{1}+\pi) \otimes HZ(\alpha_{2}+\pi)>=U(2) \otimes U(2)$ for generically chosen $\alpha_1$ and $\alpha_2$, with $<S>$ denoting the group generated by set $S$.}. Note that the proof we found is for angles irrationally related to $\pi$, however it extends to instances of angles rationally related to $\pi$. 
This is due to the fact that these angles give $U_1$ and $U_2$ whose eigenvalues have arguments irrationally related to $\pi$ or eigenvalues equal to 1  \footnote{more precisely some integer powers of $U_1$ and $U_2$ give these eigenvalues.}, thereby fulfilling the requirements in \cite{lloyd,barenco}. 
The proof for any $n=2^\gamma$ can be extended by induction from the $\gamma=1$ case, using the same methods,  while noting that an element of $\mathcal{U_{CGEN}}$ in this case can be written as $U=CZ_{\frac{n}{2},\frac{n}{2}+1}( A \otimes B)$ where $$A \otimes 1_{2 \times 2}^{\otimes \frac{n}{2}}=CZ_{1,2}...CZ_{\frac{n}{2}-2,\frac{n}{2}-1}(HZ(\alpha_1+m_1\pi) \otimes...\otimes HZ(\alpha_{\frac{n}{2}}+m_{\frac{n}{2}}\pi)), $$ and  $$ 1_{2 \times 2}^{\otimes \frac{n}{2}} \otimes B=CZ_{\frac{n}{2}+1,\frac{n}{2}+2}...CZ_{n-1,n}(HZ(\alpha_{\frac{n}{2}+1}+m_{\frac{n}{2}+1}\pi) \otimes...\otimes HZ(\alpha_{n}+m_n\pi)).$$ \\ $ A,B \in U(d)=U(2^{\frac{n}{2}})$, and $CZ_{\frac{n}{2},\frac{n}{2}+1} \in U(2^n)=U(d^2)$.

\section{Conclusions}

In this work we have relaxed the strict conditions on the sets of unitaries used for generating $t$-designs. 
This relaxation has natural relevance when considering $t$-designs derived from measurements on graph states - i.e. in the MBQC regime. We further showed that such constructions can also be used for providing new and interesting candidates for architectures demonstrating quantum speedup.

Using these techniques we have provided explicit constructions of a regular graph, such that measuring on fixed angles generates efficient $t$-designs, and classically hard to sample distributions demonstrating quanutm speedup. These techniques and graph state architectures open up more opportunities for developing and demonstrating new and simple speedup architectures. In addition, the well developed verification techniques for graph states~\cite{juan,damianalexandra,gao,cramer,serfling} provide a natural path for verification.
Moreover, graph states are broad resource across quanutm information in netwoks including
computing~\cite{briegelrauss}, fault tolerance~\cite{RHG07}, cryptographic multiparty protocols~\cite{Secret sharing Markham Sanders 08}. Indeed, the same graph state gadgets used here are universal for quantum computation~\cite{briegelrauss} and can be used to distill optimal resources for quantum metrology \cite{friis2017}. 
In this context, our results lend themselves to the integration of these ideas into future quantum networks.\\
An open question is whether the $O(n^3t^{12})$ bound on efficiency of $t$-designness shown here can be enhanced to the (optimal in $n$) bounds in \cite{brandao,peterdamian,mezher}. Another open question would be an analytical demonstration of efficiency of $t$-designness for cluster state gadgets with almost any assignment of non-adaptive fixed $XY$ angle measurements.

\bigskip

\section{Acknowledgements}
We thank Juan Bermejo-Vega for useful discussions, and for pointing out that our graph gadgets are hard to sample from classically. The authors would like to acknowledge the National Council for Scientific Research of Lebanon (CNRS-L) and the Lebanese University (LU) for granting a doctoral fellowship to R. Mezher. The authors acknowledge support from the grant VanQuTe.




\section{Appendix}
\subsection{Comment on $Conjecture$ $A$}
At some point in the \textbf{Main Results} section, we mentioned that if $Conjecture$ $A$ is true, then we can use techniques from Theorem \ref{th1} to prove that $n$-qubit cluster state gadgets $LG_{block(B^k)}$ effectively give rise to efficient $t$-designs for almost all choices of 2-qubit cluster state gadgets $G_B$. In what follows, we illustrate how this can be done for the particular version of $Conjecture$ $A$ suggested by our numerics - which are performed on 1-qubit and 2-qubit cluster state gadgets. Namely that the subdominant eigenvalue $|\lambda|$ of $M_t[\mu_B]$ is upper bounded by a constant independent of $t$ for almost all 2-qubit cluster state gadgets $G_B$. This version of Conjecture $A$ is inspired from our numerics, as well as from the result of \cite{bourgaingambourd} which showed that $|\lambda|$ is upper bounded by a constant independent of $t$ when the universal set is invertible and composed of algebraic entries, and also from the results of \cite{hastings} which showed a $|\lambda|$ upper bounded by a constant independent of $t$ (up to large values of $t$ scaling with the dimension of the unitaries) for finite gate sets chosen from the Haar measure.
In other words, if the above version of $Conjecture$ $A$ is true, then as a direct corollary 

\begin{lemma}
\label{new1}
$B$ is an $(\eta,\infty)$-TPE with $\eta \sim |\lambda| \leq C < 1$, and $C$ is indpendent of $t$.
\end{lemma}
Now, replacing Lemma \ref{lem1} in the proof of Theorem \ref{th1} by Lemma \ref{new1}, then performing the exact same steps as in the proof of Theorem \ref{th1} allows us to obtain the required result. Then, the corresponding statement for gadgets $LG_{block(B^k)}$ follows straightforwardly from the translation to MBQC developped in previous sections. 

As a final remark, if $conjecture$ $A$ is true, we would not require  $\mathcal{U_B}$ to be composed of unitaries with algebraic entries in our proofs anymore. The only reason we require unitaries with algebraic entries is to use techniques in \cite{bourgaingambourd,brandao} in order to arrive at Lemma \ref{lem1}.

\subsection{Proof of example sampling from a partially invertible set}
For simplicity, let $\alpha=\frac{\pi}{6}$ and $\beta=acos(\sqrt{\frac{1}{3}})$.
The graph gadget $G_B$ in the example of Figure \ref{fig3} gives rise to a random unitary ensemble $\mathcal{U_B}$ with random unitaries of the form 
\begin{equation} \nonumber
U_m=(HZ(\beta+m_8\pi) \otimes HZ(\beta+m_7\pi))CZ(HZ(m_6\pi)HZ(\alpha+m_5\pi)HZ(m_4\pi)\otimes HZ(\alpha+m_3\pi)HZ(m_2\pi)HZ(\alpha+m_1\pi)),    
\end{equation}
where $m_i \in \{0,1\}$ for $i=1,...,8$. Let 
\begin{equation} \nonumber
    U_m=B_m.A_m,
\end{equation}
where 
\begin{align}
B_m&=HZ(\beta+m_8\pi) \otimes HZ(\beta+m_7\pi),\nonumber \\
A_m&=CZ(HZ(m_6\pi)HZ(\alpha+m_5\pi)HZ(m_4\pi)\otimes HZ(\alpha+m_3\pi)HZ(m_2\pi)HZ(\alpha+m_1\pi)). \nonumber
\end{align}
Brute force calculation shows that $\mathcal{U_B}$ is partially invertible (up to a global phase). What remains to be shown is that $\mathcal{U_B}$ is universal. This amounts to showing that products of unitaries $U_m, U^{'}_m \in \mathcal{U_B}$ can generate any unitary in $U(4)$, in line with the results of \cite{lloyd,barenco}. Thus, as for Theorem \ref{th3}, we will show that \\
$(I)$ the Hermitian matrices $\frac{log(U_m)}{i}$ are elements of $Lie(U(4))$, and \\
$(II)$ that eigenvalues of integer multiples $U^k_m$ of $U_m$ have eigenvalues with arguments irrationally related to $\pi$.

For $(II)$, notice that $det(U_m)=e^{i(-4\beta+r\pi)}$, where $r$ is a rational number. Then, at least one of the eigenvalues $e^{i\theta}$ of $U_m$ has $\theta$ irrationally related to $\pi$, since $\beta$ is iraationally related to $\pi$. This means that for some integer $k$, $V=U^{k}_m$ has eigenvalues 1 or eigenvalues with arguments irrationally related to $\pi$. Then, for all real numbers $\lambda$, there exists an integer $m$ such that $V^m = V^{\lambda + O(1)}$, fulfilling one of the two requirements in \cite{lloyd,barenco}. 

$(I)$ follows straightforwardly from techniques in Theorem \ref{th3}. $\frac{log(B_m)}{i}$ is a general element of $Lie(U(2) \otimes U(2))$ by Lemma \ref{l1th3}, since $\beta$ is an angle irrationally related to $\pi$.  Furthermore, $A_m$ is an entangling gate not expressible as a single product of 1-qubit gates, which means that $\frac{log(B_mA_m)}{i}$ is a general element of $Lie(U(4))$ by Lemma 6.1 in \cite{brylinski}. Note that a multitude of other sets of angles $\alpha$ and $\beta$ we tested also gave partially invertible universal sets. The choice of elements uniformly at random from this set is due to the uniform probability of the different measurement results to occur.

{}

\end{document}